\documentclass[reqno,a4paper,11pt]{article}
\pdfoutput=1
\usepackage{color}

\usepackage{graphicx}
\usepackage[textwidth = 430 pt, textheight = 630 pt]{geometry}

\definecolor{MyDarkBlue}{rgb}{0.15,0.25,0.45}
\usepackage{epsfig,rotating}
\usepackage{amsmath,amssymb}
\usepackage{amsfonts}
\usepackage{mathrsfs}
\usepackage{bbm}
\usepackage[normalem]{ulem}
\usepackage[all,knot]{xy}
\xyoption{arc}

\usepackage{latexsym}
\usepackage{amsthm}

\usepackage[utf8x]{inputenc}

\usepackage{hyperref}
\hypersetup{
hypertexnames=false,
colorlinks=true,
citecolor=MyDarkBlue,
linkcolor=MyDarkBlue,
urlcolor=MyDarkBlue,
pdfauthor={Patricia Ritter and Christian S\"amann},
pdftitle={Automorphisms of Strong Homotopy Lie Algebras of Local Observables},
pdfsubject={hep-th math-ph}
breaklinks=true
}

\usepackage{tikz}
\usepackage{mathtools}

\let\fn\footnote
\renewcommand{\footnote}[1]{\linespread{1.1}\fn{#1}\linespread{1.29}}

\makeatletter\renewcommand{\section}{\@startsection
{section}{1}{\z@}{-3.5ex plus -1ex minus
    -.2ex}{2.3ex plus .2ex}{\bf }}
\makeatletter\renewcommand{\subsection}{\@startsection{subsection}{2}{\z@}{-3.25ex
plus -1ex minus
   -.2ex}{1.5ex plus .2ex}{\bf }}
\makeatletter\renewcommand{\subsubsection}{\@startsection{subsubsection}{3}{-2.45ex}{-3.25ex
plus -1ex minus -.2ex}{1.5ex plus .2ex}{\it }}
\renewcommand{\thesection}{\arabic{section}}
\renewcommand{\thesubsection}{\arabic{section}.\arabic{subsection}}
\renewcommand{\@seccntformat}[1]{\@nameuse{the#1}.~~}

\renewcommand{\theequation}{\thesection.\arabic{equation}}
\makeatletter \@addtoreset{equation}{section}

\usepackage[toc,page]{appendix}

\newtheorem{thm}{Theorem}[section]
\renewcommand{\thethm}{\thesection.\arabic{thm}}
\newtheorem{lemma}[thm]{Lemma}
\newtheorem{definition}[thm]{Definition}
\newtheorem{theorem}[thm]{Theorem}
\newtheorem{proposition}[thm]{Proposition}

\newtheorem{remark}[thm]{Remark}
\newtheorem{example}[thm]{Example}

\renewcommand{\appendices}{
\section*{Appendix}\label{appendices}\setcounter{subsection}{0}
\addcontentsline{toc}{section}{Appendix}
\setcounter{equation}{0}
\makeatletter
\renewcommand{\theequation}{\Alph{subsection}.\arabic{equation}}
\renewcommand{\thesubsection}{\Alph{subsection}}
\renewcommand{\thethm}{\Alph{subsection}.\arabic{thm}}
\@addtoreset{equation}{subsection}
\@addtoreset{thm}{subsection}
\makeatother
}

\newcommand{\myxymatrix}[1]{\vcenter{\vbox{\xymatrix{#1}}}}

\def\slasha#1{\setbox0=\hbox{$#1$}#1\hskip-\wd0\hbox to\wd0{\hss\sl/\/\hss}}

\def\periodb#1{\setbox0=\hbox{$#1$}#1\hskip-\wd0\hbox to\wd0{-}}

\newcommand{\lsb}{[\hspace{-0.05cm}[}
\newcommand{\rsb}{]\hspace{-0.05cm}]}

\newcommand{\id}{\mathrm{id}}   			

\newcommand{\CA}{\mathcal{A}}    			

\newcommand{\CC}{\mathcal{C}}

\newcommand{\CD}{\mathcal{D}}

\newcommand{\CL}{\mathcal{L}}
\newcommand{\CM}{\mathcal{M}}

\newcommand{\CV}{\mathcal{V}}

\newcommand{\CE}{\mathcal{E}}

\newcommand{\frX}{\mathfrak{X}}

\newcommand{\FR}{\mathbbm{R}}     			
\newcommand{\FC}{\mathbbm{C}}     			
\newcommand{\NN}{\mathbbm{N}}     			
\newcommand{\RZ}{\mathbbm{Z}}     			

\newcommand{\dd}{\mathrm{d}}     			
\newcommand{\dpar}{\partial}     			
\newcommand{\eps}{{\varepsilon}}			

\newcommand{\eand}{{\qquad\mbox{and}\qquad}}     		

\newcommand{\der}[1]{\frac{\dpar}{\dpar #1}}   		

\newcommand{\sU}{\mathsf{U}}     			

\newcommand{\sG}{\mathsf{G}}

\newcommand{\sL}{\mathsf{L}}

\newcommand{\sEnd}{\mathsf{End}\,}

\def\tyng(#1){\hbox{\tiny$\yng(#1)$}}			
\def\tyoung(#1){\hbox{\tiny$\young(#1)$}}			

\newcommand{\beq}{\begin{eqnarray}}
\newcommand{\eeq}{\end{eqnarray}}

\newcommand{\CatDiff}{\mathsf{Diff}}

\newcommand{\CatLinftyAlg}{\mathsf{L_\infty Alg}}

\newenvironment{myitemize}{
\vspace{-2mm}\begin{itemize}
\setlength{\itemsep}{-1mm}
}{\vspace{-2mm}\end{itemize}}

\begin{document}
\begin{titlepage}
\begin{flushright}
 DIFA 2015\\
 EMPG--15--08
\end{flushright}
\vskip 2.0cm
\begin{center}
{\LARGE \bf Automorphisms of Strong Homotopy \\[0.4cm] Lie Algebras of Local Observables}
\vskip 1.5cm
{\Large Patricia Ritter$^a$ and Christian S\"amann$^b$}
\setcounter{footnote}{0}
\renewcommand{\thefootnote}{\arabic{thefootnote}}
\vskip 1cm
{\em${}^a$ Dipartimento di Fisica ed Astronomia \\
Universit\`a di Bologna and INFN, Sezione di Bologna\\
Via Irnerio 46, I-40126 Bologna, Italy
}\\[0.5cm]
{\em ${}^b$ Maxwell Institute for Mathematical Sciences\\
Department of Mathematics, Heriot-Watt University\\
Colin Maclaurin Building, Riccarton, Edinburgh EH14 4AS,
U.K.}\\[0.5cm]
{Email: {\ttfamily patricia.ritter@bo.infn.it , C.Saemann@hw.ac.uk}}
\end{center}
\vskip 1.0cm
\begin{center}
{\bf Abstract}
\end{center}
\begin{quote}
There is a well-established procedure of assigning a strong homotopy Lie algebra of local observables to a multisymplectic manifold which can be regarded as part of a categorified Poisson structure. For a 2-plectic manifold, the resulting Lie 2-algebra is isomorphic to a sub Lie 2-algebra of a natural Lie 2-algebra structure on an exact Courant algebroid. We generalize this statement to arbitrary $n$-plectic manifolds and study automorphisms on the arising Lie $n$-algebras. Our observations may be useful in studying the quantization problem on multisymplectic manifolds.
\end{quote}
\end{titlepage}

\tableofcontents

\section{Introduction and results}

The quantization of multisymplectic manifolds \cite{Cantrijn:1999aa} is a long-standing problem. Its study is motivated by string and M-theory, where quantized multisymplectic manifolds are indeed believed to appear, see e.g.\ \cite{Saemann:2012ex} and references therein. On the other hand, the quantization of multisymplectic manifolds should be a categorified notion of quantization of symplectic manifolds. The full understanding of a mathematical concept should also include an understanding of its categorified analogues.

In this context, it is particularly important to have a detailed
picture of the classical structures behind multisymplectic manifolds
such as the (higher) Lie algebra of local observables. It was first
suggested in \cite{Baez:2008bu} that in the case of 2-plectic
manifolds, which are manifolds endowed with a closed and
non-degenerate 3-form, the observables form a semistrict Lie 2-algebra
consisting of functions and Hamiltonian one-forms. This Lie 2-algebra
has many desirable properties. First, it is given by a central
extension of the Lie algebra of Hamiltonian vector fields, just as an
ordinary Poisson algebra. Second, string theory considerations suggest
that for a 2-plectic manifold foliated by symplectic hyperplanes, a
nice reduction procedure from 2-plectic to symplectic geometry should
exist, which is the case here, cf.\ e.g.\ \cite{Ritter:2013wpa}. Finally, recall that the symplectic form of quantizable manifolds defines a principal $\sU(1)$-bundle $P$, and the Lie algebra of observables embeds into the Lie algebra of global sections of the Atiyah algebroid of $P$. There is an analogous statement for the Lie 2-algebra of local observables. Here, this Lie 2-algebra embeds into the Lie 2-algebra of global sections of a Courant algebroid corresponding to the $\sU(1)$-bundle gerbe associated with the 2-plectic manifold \cite{Baez:2008bu}.

The assignment of a categorified Lie algebra of local observables can be generalized to arbitrary multisymplectic manifolds \cite{Rogers:2010nw}, see also \cite{Fiorenza:1304.6292}. In particular, we obtain a semistrict Lie $n$-algebra of local observables from an $n$-plectic manifold. However, the embedding into higher analogues of Courant algebroids seems to have remained open until now \cite{Zambon:2010ka}. A first goal of our paper is to fill this gap. The Courant algebroid corresponding to the $\sU(1)$-bundle gerbe associated to a 2-plectic manifold $(M,\varpi)$ is a special symplectic Lie 2-algebroid with underlying vector bundle $TM\oplus T^*M$ and twisted by the 2-plectic form $\varpi$. In fact, one can build a symplectic Lie $n$-algebroid for each $n$-plectic manifold $(M,\varpi)$ containing the vector bundle $TM\oplus \wedge^{n-1} T^*M$, which we call Vinogradov algebroid. This algebroid can be twisted by $n$-plectic forms, and the twisted cases $n=1$ and $n=2$ correspond precisely to the Atiyah and the Courant algebroids. One can now show that the twisted Vinogradov algebroid comes with an associated Lie $n$-algebra containing a sub Lie $n$-algebra which is weakly isomorphic to the Lie $n$-algebra of local observables on $(M,\varpi)$.

The fact that we needed a weak isomorphism to embed the strong homotopy Lie algebra of local observables, or {\em shlalo} for short, points to a second issue. The assignment of a shlalo to an $n$-plectic manifold $(M,\varpi)$ is functorial in the sense that strict automorphisms of Lie $n$-algebras will turn the shlalo into the shlalo of an $n$-plectic manifold diffeomorphic to $(M,\varpi)$. This, however, is not the case for weak automorphisms. In particular, the weak automorphism used in embedding the shlalo into the Lie 2-algebra of sections of the Courant algebroid yields a Lie 2-algebra that is readily seen not to be the shlalo of any 2-plectic manifold.

This issue can be cured by adding structures to the $n$-plectic manifolds which transform non-trivially under weak automorphisms of Lie $n$-algebras and participate in the assignment of a shlalo to such enhanced $n$-plectic manifolds. These structures take the form of $k$-ary graded antisymmetric, linear brackets of degree $1-k$ acting on $k$-tuples of differential forms.

A similar issue appears in the assignment of a Lie $n$-algebra to a symplectic Lie $n$-algebroid. Weak automorphisms of the Lie $n$-algebra do not yield associated Lie $n$-algebras of symplectic Lie $n$-algebroids. Again, a fix for this issue is to add analogous brackets to the symplectic Lie $n$-algebroid which transform appropriately. 

Our generalization of multisymplectic manifolds to enhanced multisymplectic manifolds with brackets has a number of advantages. First, it seems to give a nice interpretation of the categorification of a symplectic manifold, at least in the case $n=2$. Here, an enhanced 2-plectic manifold $(M,\varpi,\langle-,\ldots,-\rangle)$ is, under mild restrictions, a symplectic Lie $1$-algebroid $T[1]M$. This is the natural categorification of a symplectic manifold, which can be regarded as a symplectic Lie $0$-algebroid, cf.\ \cite{Severa:2001aa}. Moreover, a Lie 1-algebroid can be regarded as a categorified manifold or 2-space\footnote{That is, a category internal to the category of smooth manifolds.} $T[1]M\rightrightarrows M$. Finally, the underlying 2-vector space of the shlalo is precisely the algebra of functions of degree 0 on the symplectic Lie $1$-algebroid $T[1]M$.

Another advantage is that the brackets provide some retrospective motivation for the quantization procedure of 2-plectic manifolds proposed in \cite{DeBellis:2010pf}, see also \cite{Barron:2014bma}. Here, 2-plectic manifolds were considered that are simultaneously symplectic manifolds. The symplectic structure was used to construct a quantum Hilbert space and the 2-plectic structure was then mapped into an additional structure on this Hilbert space. In the case of enhanced 2-plectic manifolds, certain brackets encode a symplectic form, which then can be quantized in a compatible way with the 2-plectic structure.

Finally, it seems that the brackets on multisymplectic manifolds provide us with a richer or more general higher product structure on the corresponding shlalo. Again, under mild restrictions, the resulting higher products are compatible with an additional natural structure which may be regarded as the higher analogue of the associative product in an ordinary Poisson algebra. 

This paper is structured as follows. We review $L_\infty$-algebras and their automorphisms in section 2. We then discuss the assignment of a strong homotopy Lie algebra of local observables to a multisymplectic manifold in section 3. There, we also extend the picture to enhanced multisymplectic manifolds. Section 4 gives a brief review of symplectic Lie $n$-algebroids, before the embedding of shlalos into the associated Lie $n$-algebras of symplectic Lie $n$-algebroids is proved in section 4.2. Section 4.3 again completes the picture by introducing brackets, enhancing symplectic Lie $n$-algebroids. An appendix summarizes definitions related to the description of $L_\infty$-algebras in terms of differential graded coalgebras for the reader's convenience.

\section{The category $\CatLinftyAlg$}

\subsection{Strong homotopy Lie algebras}

Strong homotopy Lie algebras or $L_\infty$-algebras can be defined in a number of ways. Possibly the most familiar form is that of a graded vector space together with a set of multilinear graded antisymmetric brackets or higher products which satisfy higher Jacobi identities \cite{Lada:1992wc}. This definition is equivalent to a  differential graded cocommutative coalgebra as shown in \cite{Lada:1994mn}. If the subspaces of homogeneous grading are finite dimensional, we can take their duals and arrive at a differential graded commutative algebra. This definition of $L_\infty$-algebras is quite familiar to physicists from the formalism of BV- or BRST-quantization, cf.\ \cite{Lada:1992wc}, as well as the AKSZ-formulation of topological field theories \cite{Alexandrov:1995kv}. The language of the latter, namely N$Q$-manifolds yields probably the most concise and readily accessible definition:
\begin{definition}\label{def:NQ-manifold}
 An \uline{N$Q$-manifold} is an $\NN$-graded manifold endowed with a vector field $Q$ of degree $1$ satisfying $Q^2=0$.
\end{definition}
N$Q$-manifolds which are trivial\footnote{General N$Q$-manifolds are in one-to-one correspondence with $L_\infty$-algebroids.} in degree $0$ are in one-to-one correspondence to $L_\infty$-algebras, for which we can take degree-wise duals of the underlying graded vector spaces. Unfortunately, the $L_\infty$-algebras we will be interested in have infinite-dimensional vector subspaces of homogeneous grading and taking duals thus is problematic. We therefore have to start from the dual picture in terms of coalgebras. 

Given a $\RZ$-graded vector space $V=\oplus_i V_i$ consisting of vector subspaces $V_i$ with grading $i$, the graded vector space $V[n]:=\oplus_i V_i$ consist of the same subspaces $V_i$ with shifted grading $i+n$. Moreover, the reduced graded symmetric algebra $\bar S(V)$ of a graded vector space $V$ consists of $V\oplus (V\odot V)\oplus V^{\odot 3}\oplus \ldots$, where $\odot$ denotes the graded symmetric tensor product. Recall that $\bar S(V)$ naturally carries the structure of a cocommutative graded coalgebra. For more details on this point and related definitions, see appendix \ref{app:Definitions}.
\begin{definition}\label{def:shLa-coalgebra}
 An \uline{$L_\infty$-algebra} is a pair $(\sL,\CD)$, where $\sL$ is a graded vector space and $\CD$ is a graded coalgebra coderivation on the graded coalgebra $\bar S(\sL[-1])$ with $\CD^2=0$. An \uline{$n$-term $L_\infty$-algebra} is an $L_\infty$-algebra, in which $\sL$ is concentrated in degrees $1-n,\ldots,0$.
\end{definition}
\noindent Note that $n$-term $L_\infty$-algebras are believed to be categorically equivalent to semistrict Lie $n$-algebras. For $n=2$, this statement has been proved in \cite{Baez:2003aa}.

Projecting the image of the coderivation $\CD$ onto $\sL[-1]\subset \bar S(\sL[-1])$ yields a linear map $\CD^1:\bar S(\sL[-1])\rightarrow \sL[-1]$ which uniquely defines all of $\CD$, cf.\ appendix \ref{app:Definitions}. Furthermore, $\CD^1$ is a sum of maps $\mu_i:(\sL[-1])^{\odot i}\rightarrow \sL[-1]$ and these maps are the above mentioned higher products $\mu_i:\sL^{\wedge i}\rightarrow \sL$ with grading $2-i$. The condition $\CD^2=0$ translates to a number of higher or homotopy Jacobi identities on the higher products, amongst which are
\begin{equation}
 \mu_1\circ \mu_1=0~,~~~\mu_1\circ \mu_2=\mu_2\circ \mu_1\eand \mu_1\circ\mu_3+\mu_2\circ\mu_2+\mu_3\circ\mu_1=0~.
\end{equation}
The firsts equation implies that $\mu_1$ is a differential, and the second equation states that this differential is compatible with an antisymmetric product. The third equation yields a controlled violation of the Jacobi identity.

As an example, consider the case of a Lie 2-algebra $\sL=\sL_{-1}\oplus \sL_0$ with $x,x_{1,2,\ldots}\in\sL_0$ and $y,y_{1,2,\ldots}\in \sL_{-1}$. The concentration of $\sL$ in degrees $-1$ and $0$ and the condition $\CD^2=0$ are then equivalent to
\begin{subequations}\label{eq:homotopy_relations}
\begin{equation}
\begin{aligned}
 \mu_1(x)&=0~,~~~\mu_2(y_1,y_2)=0~,\\
 \mu_1(\mu_2(x,y))&=\mu_2(x,\mu_1(y))~,~~~\mu_2(\mu_1(y_1),y_2)=\mu_2(y_1,\mu_1(y_2))~,\\
 \mu_3(y_1,y_2,y_3)&=\mu_3(y_1,y_2,x)=\mu_3(y_1,x_1,x_2)=0~,\\
 \mu_1(\mu_3(x_1,x_2,x_3))&=-\mu_2(\mu_2(x_1,x_2),x_3)-\mu_2(\mu_2(x_3,x_1),x_2)-\mu_2(\mu_2(x_2,x_3),x_1)~,\\
 \mu_3(\mu_1(y),x_1,x_2)&=-\mu_2(\mu_2(x_1,x_2),y)-\mu_2(\mu_2(y,x_1),x_2)-\mu_2(\mu_2(x_2,y),x_1) 
\end{aligned}
\end{equation}
and
\begin{equation}
\begin{aligned}
 \mu_2(\mu_3(x_1,&x_2,x_3),x_4)-\mu_2(\mu_3(x_4,x_1,x_2),x_3)+\mu_2(\mu_3(x_3,x_4,x_1),x_2)\\
 & -\mu_2(\mu_3(x_2,x_3,x_4),x_1)=\\
 &\mu_3(\mu_2(x_1,x_2),x_3,x_4)-\mu_3(\mu_2(x_2,x_3),x_4,x_1)+\mu_3(\mu_2(x_3,x_4),x_1,x_2)\\
 &-\mu_3(\mu_2(x_4,x_1),x_2,x_3)
 -\mu_3(\mu_2(x_1,x_3),x_2,x_4)-\mu_3(\mu_2(x_2,x_4),x_1,x_3)~.
\end{aligned}
\end{equation}
\end{subequations}

\subsection{Morphisms of $L_\infty$-algebra}\label{ssec:sh-morphisms}

A na\"ive definition of morphisms between two $L_\infty$-algebras is given by chain maps between the underlying complexes. Here, however, we are interested in generalized, weak or quasi-isomorphisms of $L_\infty$-algebras. The appropriate definition of these is readily inferred if $L_\infty$-algebras are regarded as differential graded coalgebras as in definition \ref{def:shLa-coalgebra}, cf.\ e.g.\ appendix A of \cite{Fregier:2013dda}. 
\begin{definition}
 Given $L_\infty$-algebras $(\sL,\CD)$ and $(\sL',\CD')$, a \uline{(weak) morphism of $L_\infty$-algebras} $\Psi:(\sL,\CD)\rightarrow (\sL',\CD')$ is a morphism of the underlying graded coalgebras compatible with the codifferentials. That is, $\Psi\circ \CD=\CD'\circ \Psi$. 
\end{definition}
Note that we can again restrict the image of $\Psi$ to $\sL'$, yielding a map $\Psi^1$ which fully defines $\Psi$. This map is a sum $\Psi^1=\Psi^1_1+\Psi^1_2+\ldots$, where $\Psi^1_i:\sL^{\wedge i}\rightarrow \sL$ are linear maps of degree $1-i$. Clearly, a morphism of $n$-term $L_\infty$-algebras is a morphism with $\Psi^1_k=0$ for $k>n$.

Often in the literature, only strict morphisms of $L_\infty$-algebras are defined.
\begin{definition}
 A morphism of $L_\infty$-algebras $\Psi$ is called \uline{strict}, if $\Psi^1_i=0$ for $i>1$. 
\end{definition}

To translate back to the bracket formulation, apply the equation $\Psi\circ \CD=\CD'\circ \Psi$ to elements of $\sL[-1]^{\odot j}$ for each $j$. We then rewrite $\Psi$, $\CD$ and $\CD'$ in terms of $\Psi^1=\Psi^1_1+\Psi^1_2+\ldots$, $\CD^1=\mu_1+\mu_2+\ldots$ and $(\CD')^1=\mu'_1+\mu'_2+\ldots$. This results in the following relations:
\begin{equation}
 \begin{aligned}
  &\mu'_1(\Psi^1_1(\ell))=\Psi^1_1(\mu_1(\ell))~,\\
  &\mu'_2(\Psi^1_1(\ell_1),\Psi^1_1(\ell_2))=\Psi^1_1(\mu_2(\ell_1,\ell_2))\pm\Psi^1_2(\mu_1(\ell_1),\ell_2)\pm \Psi^1_2(\ell_1,\mu_1(\ell_2))+\mu'_1(\Psi^1_2(\ell_1,\ell_2))~,\\
  &\mu'_3(\Psi^1_1(\ell_1),\Psi^1_1(\ell_2),\Psi^1_1(\ell_3))=\Psi^1_1(\mu_3(\ell_1,\ell_2,\ell_3))-\mu'_1(\Psi^1_3(\ell_1,\ell_2,\ell_3))-\\
  &\hspace{4.5cm}-\left(\Psi^1_2(\ell_1,\mu_2(\ell_2,\ell_3))\pm\mu'_2(\Psi^1_1(\ell_1),\Psi^1_2(\ell_2,\ell_3))+\mbox{cycl.}\right)
 \end{aligned}
\end{equation}
etc., where the signs need to be adjusted according to the grading of the maps and arguments. Note that $\Psi^1_1$ is simply a chain map of the underlying complexes. 

We are now in a position to define the category $\CatLinftyAlg$.
\begin{definition}
 The category $\CatLinftyAlg$ has $L_\infty$-algebras as objects and weak $L_\infty$-algebra morphisms as morphisms.
\end{definition}

Let us conclude this section with a brief remark on automorphisms of
$L_\infty$-algebras which are invertible morphisms of
$L_\infty$-algebras.
\begin{lemma}
 An automorphism on an $L_\infty$-algebras $\Psi:(\sL,\CD)\rightarrow (\sL,\CD')$ is specified by an invertible chain map $\Psi^1_1:\sL\rightarrow \sL$ together with an arbitrary set of linear maps $\Psi^1_i:\sL^{\wedge i}\rightarrow \sL$ with $i>1$ of degree $1-i$.
\end{lemma}
\begin{proof}
 From the above considerations it is clear that the equation $\Psi\circ \CD=\CD'\circ \Psi$ applied to $\sL[-1]^{\odot i}$ can be solved for $\mu'_i(\Psi^1_1(\ell_1),\ldots, \Psi^1_1(\ell_i))$, which defines the higher products and thus the codifferential $\CD'$ on $\sL$ iteratively. It remains to show that $\Psi$ is invertible if $\Psi^1_1$ is. One readily sees that the composition of two morphisms $\Phi$ and $\Psi$ of $L_\infty$-algebras satisfies
 \begin{equation}
 \begin{aligned}
  (\Phi\circ \Psi)^1_1&=\Phi^1_1\circ \Psi^1_1~,\\
  (\Phi\circ \Psi)^1_i&=\Phi^1_i(\Psi^1_1(-),\ldots,\Psi^1_1(-))+R_i~,
 \end{aligned}
 \end{equation}
 where $R_i$ are terms involving $\Phi^1_j$ with $j<i$, schematically
$\Phi_j^i(\Psi_{i-j+1}^1(-,\ldots,-),\ldots,-)$. The morphism $\Psi$ is an automorphism if and only if there is a map $\Phi$ such that $(\Phi\circ \Psi)^1_1=\id$ and $(\Phi\circ \Psi)^1_i=0$ for $i>1$. If $\Psi^1_1$ is invertible, we can define $\Phi^1_1:=(\Psi^1_1)^{-1}$ and fix the $\Phi^1_i$ iteratively by putting $(\Phi\circ \Psi)^1_i=0$ for $i>1$. This yields the inverse to $\Psi$ and shows that $\Psi$ is an automorphism.
\end{proof}

\section{$L_\infty$-algebra of local observables on multisymplectic manifolds}\label{sec:shlalos}

\subsection{The Poisson algebra of a symplectic manifold}

Given a symplectic manifold $(M,\omega)$, we can associate to a smooth function $f\in \CC^\infty(M)$ a Hamiltonian vector field $X_f$ such that $\iota_{X_f}\omega=-\dd f$. This, in turn, allows us to introduce a Lie algebra structure on $\CC^\infty(M)$ by defining $\{f,g\}:=-\iota_{X_f}\iota_{X_g}\omega=:\pi(\dd f,\dd g)$, where $f,g\in \CC^\infty(M)$ and $\pi\in\Gamma(\wedge^2 TM)$ is a bivector field. This Lie algebra can be regarded as a central extension of the Lie algebra of Hamiltonian vector fields by the constant functions because $X_{\{f,g\}}=[X_f,X_g]_{TM}$. Together with the ordinary product of functions, $(\CC^\infty(M),\pi)$ forms a Poisson algebra, but since the higher analogue of this associative product remains opaque, we shall focus mostly on the Lie algebra structure.

Recall that a symplectomorphism between two symplectic manifolds $(M_1,\omega_1)$ and $(M_2,\omega_2)$ is a smooth map $\phi:M_1\rightarrow M_2$ such that $\omega_1=\phi^*\omega_2$. On the other hand, a Poisson map between two Poisson manifolds $(M_1,\pi_1)$ and $(M_2,\pi_2)$ is a smooth map $\phi:M_1\rightarrow M_2$ such that $\pi_1$ and $\pi_2$ are $\phi$-related\footnote{That is, $\dd\phi_x \pi_1(x)=\pi_2(\phi(x))$ for all $x\in M$.}. The above construction of a Poisson bivector from a symplectic form extends to a functor, embedding the category of symplectic manifolds into that of Poisson manifolds as a full subcategory. In particular, automorphisms between symplectic manifolds are evidently in one-to-one correspondence with automorphisms of Poisson manifolds. 

A Poisson algebra homomorphism $\psi:\CC^\infty(M)\rightarrow \CC^\infty(N)$ between the Poisson algebras of two Poisson manifolds is necessarily a homomorphism between the commutative algebras $\CC^\infty(M)$ and $\CC^\infty(N)$, cf.\ \cite{daSilva:1999aa}. All such homomorphisms are of the form $\psi=\phi^*$ for some smooth map $\phi:M\rightarrow N$. In particular, any Poisson algebra automorphisms $\psi$ is of the form $\psi=\phi^*$ for some diffeomorphism $\phi$. Given a Poisson manifold $(M,\pi)$ originating from a symplectic manifold $(M,\omega)$, we therefore have a commutative square
\begin{equation}
    \myxymatrix{
    (M,\omega) \ar@{->}[r] \ar@{->}[d]_{\phi} &
    (M,\pi) \ar@{->}[r] \ar@{->}[d]_{\phi} & (\CC^\infty(M),\{-,-\}_{\pi}) \ar@{->}[d]^{(\phi^{-1})^*}\\
    (M,\omega') \ar@{->}[r] & (M,\pi') \ar@{->}[r] & (\CC^\infty(M),\{-,-\}_{\pi'})
    }
\end{equation}
with $\omega'=(\phi^{-1})^*\omega$ and $\pi'=(\phi^{-1})^*\pi$.

\subsection{Local observables on multisymplectic manifolds}

Multisymplectic manifolds generalize the notion of symplectic manifolds.
\begin{definition}
 A \uline{multisymplectic manifold of degree $n+1$} is a smooth manifold endowed with a closed $n+1$-form $\varpi$ which is non-degenerate in the sense that $\iota_X\varpi=0$ for an $X\in \Gamma(TM)$ implies $X=0$. We also call such a manifold an \uline{$n$-plectic manifold}. 
\end{definition}
\noindent Canonical examples include ordinary symplectic manifolds, which are 1-plectic manifolds, and Lie groups, which are 2-plectic manifolds with the canonical left-invariant 3-form built from the Lie bracket and the Cartan-Killing form \cite{Baez:2009:aa}. Note that each orientable manifold of dimension $n$ comes with a natural $(n-1)$-plectic form given by its volume form.

\begin{definition}
 Given an $n$-plectic manifold $(M,\varpi)$, we call an element $\alpha$ of $\Omega^{n-1}(M)$ \uline{Hamiltonian} if there is a Hamiltonian vector field $X_\alpha\in\frX(M)$ such that
 \begin{equation}
  \dd\alpha=-\iota_{X_\alpha}\varpi~.
 \end{equation}
 The set of Hamiltonians and Hamiltonian vector fields will be denoted by $\Omega^{n-1}_{\rm Ham}(M)$ and $\frX_{\rm Ham}(M)$, respectively.
\end{definition}
Due to the non-degeneracy of $\varpi$, the Hamiltonian vector field is unique if it exists. Moreover, the Hamiltonian vector fields form a subalgebra of $\frX(M)$ as one readily verifies. For more details, see e.g.\ \cite{Fiorenza:1304.6292}.

The $L_\infty$-algebra of local observables is now defined as a central extension of the Lie algebra of Hamiltonian vector fields as follows \cite{Baez:2008bu,Baez:2009:aa,Rogers:2010nw}, see also \cite{Fiorenza:1304.6292}.
\begin{definition}\label{def:ordinary_shlalo}
 The \uline{strong homotopy Lie algebra of local observables (shlalo) of an $n$-plectic manifold $(M,\varpi)$} is given by the complex
 \begin{equation}
  \CC^\infty(M)\xrightarrow{~\dd~}\Omega^1(M)\xrightarrow{~\dd~}\ldots\xrightarrow{~\dd~}\Omega^{n-2}(M)\xrightarrow{~\dd~}\Omega^{n-1}_{\rm Ham}(M)
 \end{equation}
 with assigned degrees $1-n,\ldots,0$ and the only non-trivial products being
 \begin{equation}
  \mu_i(\alpha_1,\ldots,\alpha_i):=(-1)^{\binom{i+1}{2}}\iota_{X_{\alpha_1}}\ldots\iota_{X_{\alpha_i}}\varpi
 \end{equation}
 for $\alpha_1,\ldots,\alpha_n\in \Omega^{n-1}_{\rm Ham}(M)$.
\end{definition}
 
In the case of a 2-plectic manifold $(M,\varpi)$, for example, we obtain the 2-term $L_\infty$-algebra $\CC^\infty(M)\rightarrow \Omega^1_{\rm Ham}(M)$ with non-trivial products
\begin{equation}
 \mu_1(f):=\dd f~,~~~\mu_2(\alpha_1,\alpha_2):=-\iota_{X_{\alpha_1}}\iota_{X_{\alpha_2}}\varpi~,~~~\mu_3(\alpha_1,\alpha_2,\alpha_3)=\iota_{X_{\alpha_1}}\iota_{X_{\alpha_2}}\iota_{X_{\alpha_3}}\varpi
\end{equation}
for $f\in \CC^\infty(M)$ and $\alpha_i\in \Omega^1_{\rm Ham}(M)$. The explicit example of a Lie group $\sG$ with canonical 2-plectic form has been discussed in \cite{Baez:2009:aa}.

It is now clear that a diffeomorphism on $M$ induces a strict Lie 2-algebra isomorphism on the above shlalo. However, let us apply the automorphism of $L_\infty$-algebras used e.g.\ in \cite[Section 5.3]{Ritter:2013wpa} with $\Psi_1=\id$ and $\Psi_2(\alpha_1,\alpha_2):=\iota_{X_{\alpha_1}}\alpha_2-\iota_{X_{\alpha_2}}\alpha_1$. This yields a new 2-term $L_\infty$-algebra $\sL'$ with products $\mu'_i$ and we have in particular
\begin{equation}
 \mu'_2(f,\alpha)=\iota_{X_\alpha}\dd f~.
\end{equation}
Clearly, $\sL'$ is not the $L_\infty$-algebra of local observables of a 2-plectic manifold, because for such shlalos, $\mu_2(f,\alpha)$ always vanishes.

\subsection{Generalized construction}

Let us now extend the association of a shlalo to an $n$-plectic manifold $(M,\varpi)$ such that the association covers the image of automorphisms of shlalos. For this, we need to enhance the input data from an $n$-plectic manifold to an $n$-plectic manifold with additional structure.

First, let us slightly constrain the $L_\infty$-algebra automorphisms that we will admit on shlalos. Just as in the case of Poisson algebras on Poisson manifolds, we would like to respect the natural associative product on smooth functions as well as the product between smooth functions and $n$-forms. This implies that the chain map encoded in $\Psi^1_1$ is identified with the pullbacks along a diffeomorphism $\phi$ on $M$, $\Psi^1_1=\phi^*$.
\begin{definition}
 An \uline{automorphism of shlalos} is an automorphism $\Psi$ of $L_\infty$-algebras with $\Psi^1_1=\phi^*$, where $\phi$ is a diffeomorphism on the underlying manifold $M$. A \uline{strict automor\-phism of shlalos} is an automorphism $\Psi$ of $L_\infty$-algebras with $\Psi^1_i=0$ for $i>1$.
\end{definition}
\noindent We have an immediate implication for strict automorphisms.
\begin{proposition}
 Consider the shlalo $\sL$ of a 2-plectic manifold $(M,\varpi)$. A strict automorphism of shlalos $\Psi$ with $\Psi^1_1=\phi^*$ for some diffeomorphism $\phi:M\rightarrow M$ turns $\sL$ into the shlalo of $(M,\phi^*\varpi)$.
\end{proposition}

Let us now generalize the picture to extend the above proposition to automorphisms of shlalos which are not strict. For this, we have to endow $M$ with some additional structure, which transforms under the $\Psi^1_k$. 
\begin{definition}
 An \uline{($n$-)enhanced $n$-plectic manifold} $(M,\varpi,\langle-,\ldots,-\rangle_i)$ is an $n$-plectic manifold endowed with a set of graded antisymmetric brackets $\langle-,\ldots,-\rangle_i:(\Omega^{<n})^{\wedge i}\rightarrow \Omega^{<n}$ of degree $1-i$ for $2\leq i$.
\end{definition}
An important example for an enhanced $n$-plectic manifold is now the following.
\begin{example}\label{ex:3d-Poisson}
Given a three-dimensional Poisson manifold with a choice of volume form, $(M,\pi,\omega_{\rm vol})$, we obtain an enhanced 2-plectic manifold by defining the 2-plectic form and the bracket $\langle-,-\rangle:\Omega^1(M)\wedge \Omega^1(M)\rightarrow \CC^\infty(M)$ according to
\begin{equation}
 \varpi:=\dd \mu_{\rm vol}\eand \langle\alpha,\beta\rangle:=\pi(\alpha,\beta)~.
\end{equation}
\end{example}
\noindent We will return to this and similar examples later.

Note that there is a bijection between $L_\infty$-algebra morphisms $\Phi_M$ with $(\Phi_M)^1_1=\id$ and the bracket structure on an enhanced $n$-plectic manifold $M$, given by
\begin{equation}
 (\Phi_M)^1_1:=\id~,~~~(\Phi_M)^1_i:=\langle-,\ldots,-\rangle_i~.
\end{equation}
This allows us to define the action of an automorphism on an enhanced $n$-plectic manifold.
\begin{definition}
 An \uline{automorphism on an enhanced $n$-plectic manifold} 
 \begin{equation}
  \Psi:(M,\varpi,\langle-,\ldots,-\rangle_i)\rightarrow(M,\varpi',\langle-,\ldots,-\rangle'_i)
 \end{equation}
 is a diffeomorphism $\phi:M\rightarrow M$ together with a set of maps $\Psi_k:(\Omega^{<n})^{\wedge k}\rightarrow \Omega^{<n}$ of degree $1-k$ such that
 \begin{equation}
  \omega=\phi^*\omega\eand \langle-,\ldots,-\rangle'_i=(\Psi_k\circ\Phi_M)_i\circ (\phi^*\wedge \ldots \wedge \phi^*)~.
 \end{equation}
\end{definition}
Since the automorphism acts in the same way as in $\CatLinftyAlg$, composition of automorphisms is associative.

\begin{definition}
 The \uline{shlalo of an enhanced $n$-plectic manifold} $(M,\varpi,\langle-,\ldots,-\rangle_i)$ is the $L_\infty$-algebra obtained by applying the automorphism $\Psi$ of $L_\infty$-algebras with $\Psi^1_1=\id$ and $\Psi^1_k=\langle-,\ldots,-\rangle_k$ to the shlalo of the $n$-plectic manifold $(M,\varpi)$ from definition \ref{def:ordinary_shlalo}.
\end{definition}

Let us consider again the example $n=2$. Here, an enhanced 2-plectic manifold $M$ comes with a 2-plectic form $\varpi$ and an antisymmetric map $\langle-,-\rangle:\Omega^1\times \Omega^1\rightarrow \CC^\infty(M)$. The resulting shlalo reads as
\begin{equation}\label{eq:example-extended-shlalo-2-plectic}
\begin{aligned}
 \mu_1(f):=\dd f~,~~~\mu_2(\alpha_1,\alpha_2):=-\iota_{X_{\alpha_1}}\iota_{X_{\alpha_2}}\varpi-\dd\langle\alpha_1,\alpha_2\rangle~,~~~
 \mu_2(\alpha,f):=-\langle\alpha,\dd f\rangle~,\\
 \mu_3(\alpha_1,\alpha_2,\alpha_3):=\iota_{X_{\alpha_1}}\iota_{X_{\alpha_2}}\iota_{X_{\alpha_3}}\varpi-\big(\langle\alpha_1,\iota_{X_{\alpha_2}}\iota_{X_{\alpha_3}}\varpi\rangle+\mu_2(\alpha_1,\langle\alpha_2,\alpha_3\rangle)+\rm{cycl.}\big)~.
\end{aligned}
\end{equation}

We now have obviously the following statement, which was our first goal:
\begin{proposition}
 The application of an automorphism $\Psi$ of shlalos on the shlalo of an enhanced $n$-plectic manifold yields the shlalo of another enhanced $n$-plectic manifold.
\end{proposition}

Let us finally remark that our construction generalizes to $n$-enhanced manifolds, in which the $n$-plectic structure is taken to be
zero.
\begin{definition}
 An \uline{$n$-enhanced manifold} $(M,\langle-,\ldots,-\rangle_i)$ is a manifold endowed with a set of graded antisymmetric brackets $\langle-,\ldots,-\rangle_i:(\Omega^{<n})^{\wedge i}\rightarrow \Omega^{<n}$ of degree $1-i$ for $2\leq i$. The \uline{shlalo of an $n$-enhanced manifold} $(M,\langle-,\ldots,-\rangle_i)$, $i\leq n$, is the $L_\infty$-algebra obtained by applying the automorphism $\Psi$ of $L_\infty$-algebras with $\Psi^1_1=\id$ and $\Psi^1_k=\langle-,\ldots,-\rangle_k$ to the trivial $n$-term $L_\infty$-algebra 
 \begin{equation}
  \sL^0_M:=\CC^\infty(M)\xrightarrow{~\dd~}\Omega^1(M)\xrightarrow{~\dd~}\ldots\xrightarrow{~\dd~}\Omega^{n-1}_{\rm Ham}(M)
 \end{equation}
with all higher products $\mu_i=0$ for $i>1$.
\end{definition}
As we will show in the next section, many manifolds with structures arising in Mathematical Physics give rise to enhanced manifolds.

\subsection{Examples: Symplectic Lie $n$-algebroids}

Let us begin with examples of enhanced manifolds, neglecting the $n$-plectic structure. We started from the Poisson algebra associated to a symplectic manifold. It is well-known that symplectic manifolds can be regarded as symplectic Lie 0-algebroids, cf.\ \cite{Severa:2001aa} or \cite{Roytenberg:0203110}. Going up in the ladder of categorification, we arrive at symplectic Lie 1-algebroids, which can be identified with Poisson manifolds $(M,\pi)$.

As a Lie 1-algebroid, the Poisson manifold $(M,\pi)$ is given by the manifold $T[1]M$, which can be regarded as a 2-space. Recall that a smooth 2-space is a category internal to $\CatDiff$\footnote{the category of smooth manifolds and smooth maps between them}, that is, a category in which both objects and morphisms are smooth manifolds and the source, target, identity and composition maps are smooth. In the case of $T[1]M$, the objects are $M$ and the morphisms are $TM$. The source and target maps are both identified with the bundle projection and the identity map $\id:x\mapsto T_xM$ yields the zero vector in $T_xM$. Composition of morphisms is only defined for morphisms over the same base point: $(x,a)\circ (x,b):=(x,a+b)$.

Just as expected, we are now quantizing a 2-space instead of a
1-space. Moreover, the functions with values of degree $0$ are indeed given by
$\CC^\infty(M)\oplus \Gamma(T^*[1]M)$. Here, elements of $\Gamma(T^*[1]M)\cong \Omega^1(M)$ are one-forms with grading -1, compensating the grading +1 of the vector fields. More explicitly in local coordinates $x^i$ on $M$ and $\xi^i$ in the fibers of $T[1]M$, a function reads as $f(x)=f^{(0)}(x)+f^{(-1)}_i(x)\xi^i$, where $f^{(-1)}_i$ can be identified with the components of the one-form $f^{(-1)}_i\dd x^i$ with grading -1. This is the vector space underlying the 2-term $L_\infty$-algebra constructed in the following.

Regarding $(M,\pi)$ as an enhanced manifold, we obtained the 2-term $L_\infty$-algebra\linebreak $\CC^\infty(M)\rightarrow\Omega^1(M)$ with products
\begin{equation}
\begin{aligned}
 \mu_1(f):=\dd f~,~~~\mu_2(\alpha_1,\alpha_2):=-\dd\pi(\alpha_1,\alpha_2)~,~~~
 \mu_2(\alpha,f):=-\pi(\alpha,\dd f)~,\\
 \mu_3(\alpha_1,\alpha_2,\alpha_3):=\pi(\alpha_1,\dd\pi(\alpha_2,\alpha_3))+\pi(\alpha_2,\dd\pi(\alpha_3,\alpha_1))+\pi(\alpha_3,\dd\pi(\alpha_1,\alpha_2))~.
\end{aligned}
\end{equation}
Note that even though $\pi$ is a Poisson tensor, $\mu_3(\alpha_1,\alpha_2,\alpha_3)$ does not necessarily vanish. Interestingly, the above products respect the associative product between functions in the sense that
\begin{equation}
 \mu_1(fg)=f\mu_1(g)+\mu_1(f)g\eand \mu_2(\alpha,fg)=f\mu_2(\alpha,g)+\mu_2(\alpha,f)g~.
\end{equation}
Unfortunately, this does not extend to the associative product between functions and one-forms. However, upon closer inspection it trivially extends to the shlalos of enhanced manifolds, if the brackets are given by multivector fields.
\begin{proposition}
 Let $(M,\langle-,\ldots,-\rangle_i)$ be an $n$-enhanced manifold with or without $n$-plectic structure for which the brackets are encoded in multivector fields, i.e.\ $\langle-,\ldots,-\rangle_i=\xi_{i,q}(-,\ldots,-)$, $\xi_{i,q}\in \Gamma(\wedge^q TM)$. Then the corresponding shlalo respects the associative product on $\CC^\infty(M)$ in the following sense:
 \begin{equation}
  \mu_n(A_1,\ldots,fg,\ldots,A_n)=f\mu_n(A_1,\ldots,g,\ldots,A_n)+\mu_n(A_1,\ldots,f,\ldots,A_n)g~,
 \end{equation}
 where $A_i\in \Omega^{<p}$.
\end{proposition}
\begin{proof}
 First, note that $\mu_i(-,\ldots,f,\ldots,-)=0$ for $i>1$ in the shlalo for an ordinary
 $n$-plectic manifold. Enhancing the manifold, the only non-trivial products involving functions are $\mu_1(f):=\dd f$ and
 $\mu_2(A,f):=\xi_{2,p}(A,\dd f)$, where $A\in \Omega^{p-1}$ and
 $\xi_{2,n}\in \Gamma(\wedge^p TM)$. For both products, the statement clearly holds.
\end{proof}

Note that automorphisms of Lie 1-algebroids have to preserve the fact that $\langle-,-\rangle$ originates from a Poisson bivector. Correspondingly, the general automorphisms of an $n$-enhanced manifold restrict to a subset. This is to be expected, as more structure is being preserved. 

\

Next in the hierarchy of symplectic Lie $n$-algebroids is the case $n=2$, which corresponds to a Courant algebroid.
\begin{definition}
 A \uline{Courant algebroid structure} on a vector bundle $E\rightarrow M$ over a smooth manifold $M$ is given by a non-degenerate bilinear form $(-,-):\Gamma(E)\times \Gamma(E)\rightarrow \CC^\infty(M)$, an antisymmetric bracket $\lsb-,-\rsb:\Gamma(E)\wedge\Gamma(E)\rightarrow \Gamma(E)$ and a bundle map $\rho:E\rightarrow TM$ called the \uline{anchor map}. From these maps, we derive the following additional maps:
 \begin{equation}
 \begin{aligned}
  J(e_1,e_2,e_3)&:=\lsb\lsb e_1,e_2\rsb,e_3\rsb+\lsb\lsb e_2,e_3\rsb,e_1\rsb+\lsb\lsb e_3,e_1\rsb,e_2\rsb~,\\
  T(e_1,e_2,e_3)&:=-\tfrac13 \left(\big(\lsb e_1,e_2\rsb,e_3\big)+\big(\lsb e_2,e_3\rsb,e_1\big)+\big(\lsb e_3,e_1\rsb,e_2\big)\right)~,\\
  \big( \CD f,e\big) &:=\tfrac12 \rho(e)f
 \end{aligned}
 \end{equation}
 for $e,e_{1,2,3}\in \Gamma(E)$ and $f\in\CC^\infty(M)$, which have to satisfy the following properties
 \begin{myitemize}
  \item[(i)] $\rho\lsb e_1,e_2\rsb=[\rho(e_1),\rho(e_2)]$;
  \item[(ii)] $-J(e_1,e_2,e_3)=\CD T(e_1,e_2,e_3)$
  \item[(iii)] $\lsb e_1, f e_2\rsb=f\lsb e_1,e_2\rsb+(\rho(e_1)f)e_2-\big( e_1,e_2\big) \CD f$;
  \item[(iv)] $\big(\CD f_1,\CD f_2\big) =0$;
  \item[(v)] $\rho(e_1)\big( e_2,e_3\big)=\big(\lsb e_1,e_2\rsb+\CD\big( e_1,e_2\big),e_3\big)+\big(e_2,\lsb e_1,e_3\rsb+\CD\big( e_1,e_3\big)\big)$
 \end{myitemize}
 again for all $e,e_{1,2,3}\in \Gamma(E)$ and $f,f_{1,2}\in\CC^\infty(M)$.
\end{definition}
An important example is now the following.
\begin{example}\label{ex:exact_Courant_algebroid}
Given a manifold $M$ endowed with a closed 3-form $\varpi$, we have an exact\footnote{i.e.\ there is an exact sequence of vector bundles $T^*M\xrightarrow{~\rho^*~}TM\oplus T^*M\xrightarrow{~\rho~} TM$} Courant algebroid given by the vector bundle $E=TM\oplus T^*M$ together with the maps
\begin{equation}
\begin{aligned}
 \rho(X+\alpha)&:=X~,\\
 (X+\alpha,Y+\beta)&:=\tfrac12(\iota_X \beta+\iota_Y\alpha)~,\\
 \lsb X+\alpha,Y+\beta\rsb&:=[X,Y]+\CL_X\beta-\CL_Y \alpha+\tfrac12(\dd \iota_Y \alpha-\dd \iota_X \beta)+\iota_X\iota_Y\varpi~,
\end{aligned}
\end{equation}
where $X,Y\in \frX(M)$ and $\alpha,\beta\in \Omega^1(M)$. As a consequence, we have $\CD f=0+\dd f$ for $f\in \CC^\infty(M)$.
\end{example}
Note that any Courant algebroid $E\xrightarrow{\rho} TM\rightarrow M$ can be regarded as a 3-space (a strict 2-category internal to $\CatDiff$) $E\rightrightarrows TM\rightrightarrows M$. The source and target maps of 2-morphisms $E\rightarrow TM$ are identified with the anchor map, and the source and target maps of 1-morphisms are the bundle projection. Since the algebra of functions we are going to be interested in is given by $\CC^\infty(M)\times \Omega^1(M)\times \Omega^2(M)$, we choose $E=\wedge^2 TM\oplus TM$. In the special case of an orientable manifold $M$ of dimension $3$, we can use the volume form to identify $\wedge^2 TM$ with $T^*M$ and $E=TM\oplus T^*M$ agrees at least as a vector bundle with the Courant algebroid of example \ref{ex:exact_Courant_algebroid}.

The functions of degree $0$ on this Courant algebroid are given by 
\begin{equation}
 \Gamma(\wedge^2 T^*[1]M\oplus T^*[1]M)\oplus \CC^\infty(M)\cong\Omega^2(M)\oplus \Omega^1(M)\oplus \CC^\infty(M)~.
\end{equation}
Recall that a Courant algebroid is defined by a pseudo-Euclidean inner product $(-,-)$ on $E$, an anchor map $\rho:E\rightarrow TM$ and a bracket $\lsb-,-\rsb:\Gamma(E)\wedge \Gamma(E)\rightarrow \Gamma(E)$. We then have the following natural choices for higher brackets:
\begin{equation}
 \begin{aligned}
  &\langle-,-\rangle~:~\Omega^2(M)\times \Omega^2(M)\rightarrow \Omega^1(M)~,~~~&\langle\beta_1,\beta_2\rangle&:=(\rho(\lsb\beta_1^\sharp,\beta_2^\sharp\rsb))^\flat~,\\
  &\langle-,-\rangle~:~\Omega^1(M)\times \Omega^2(M)\rightarrow \CC^\infty(M)~,~~~&\langle\alpha,\beta\rangle&:=( \alpha^\sharp,\beta^\sharp)~,\\
  &\langle-,-,-\rangle~:~\wedge^3\Omega^2\rightarrow \CC^\infty(M)~,~~~&\langle\beta_1,\beta_2,\beta_3\rangle&:=(\beta^\sharp_1,\lsb\beta^\sharp_2,\beta^\sharp_3\rsb)+{\rm cyclic}~,
 \end{aligned}
\end{equation}
where $\sharp$ and $\flat$ are the musical isomorphisms induced by the
pseudo-Euclidean structure $(-,-)$.

Alternatively, we could start from the exact Courant algebroid of example \ref{ex:exact_Courant_algebroid} and construct brackets from the contained maps.

In both cases, the resulting shlalo reads as
\begin{equation}
\begin{aligned}
 &\mu_1=\dd~,~~~\mu_2(\alpha,f)=0~,~~~\mu_2(\beta,f)=\langle \beta,\dd f\rangle~,~~~ \mu_2(\alpha,\beta)=\langle\dd \alpha,\beta\rangle-\dd \langle\alpha,\beta\rangle~,\\
 &\mu_2(\beta_1,\beta_2)=-\dd\langle\beta_1,\beta_2\rangle~,\\
 &\mu_3(\alpha,\beta_1,\beta_2)=-\langle\dd\langle\alpha,\beta_1\rangle,\beta_2\rangle+\dd\langle\langle\alpha,\beta_1\rangle,\beta_2\rangle-\langle\alpha,\dd\langle\beta_1,\beta_2\rangle\rangle+\langle\dd \alpha,\beta_1,\beta_2\rangle+\ldots~,\\
 &\mu_3(\beta_1,\beta_2,\beta_3)=-\langle\dd\langle\beta_1,\beta_2\rangle,\beta_3\rangle-\dd \langle\beta_1,\beta_2,\beta_3\rangle+\ldots~,
\end{aligned}
\end{equation}
where the $\ldots$ indicate additional terms arising from antisymmetrization of the arguments.

\subsection{Examples: Nambu-Poisson manifolds}

The archetypical example of an enhanced $2$-plectic manifold is a Poisson manifold. There is the following generalization \cite{Nambu:1973qe,Takhtajan:1993vr}, which is relevant to physics, cf.\ \cite{DeBellis:2010pf}:
\begin{definition}
 A \uline{Nambu-Poisson manifold} $(M,\pi)$ is a manifold with a multivector field $\pi\in \Gamma(\wedge^n TM)$ such that the {Nambu-Poisson bracket}
\begin{equation}
 \{f_1,\ldots,f_n\}=\pi(\dd f_1,\ldots,\dd f_n)~,~~~f_i\in \CC^\infty(M),
\end{equation}
satisfies the {\em fundamental identity}
\begin{equation}\label{eq:fundamental_identity}
\begin{aligned}
 \{f_1,\ldots,f_{n-1},\{g_1,\ldots,g_n\}\}=&\\
 \{\{f_1,\ldots,f_{n-1},&g_1\},\ldots,g_n\}\}+\ldots+ \{g_1,\ldots,\{f_1,\ldots,f_{n-1},g_n\}\}~.
\end{aligned}
\end{equation}
\end{definition}
Note that a manifold can carry many Nambu-Poisson brackets simultaneously. We can then use the resulting multivector fields to define bracket structures on these manifolds.

The example one step up from a 2-plectic Poisson manifold is a 3-plectic Nambu-Poisson manifold $(M,\varpi)$ with Nambu-Poisson tensor $\pi$ of rank $3$. A natural bracket is then defined as
\begin{equation}
 \langle-,-\rangle: \Omega^1\wedge \Omega^2\rightarrow \CC^\infty(M)~,~~~\langle \alpha,\beta\rangle=\pi^{ijk}\alpha_{ij}\beta_k~,~~~\alpha\in \Omega^2(M)~,~~\beta\in\Omega^1(M)~.
\end{equation}
The resulting shlalo of this enhanced 3-plectic manifold is based on the complex
\begin{equation}
 \CC^\infty(M)\rightarrow \Omega^1(M)\rightarrow \Omega^2_{\rm Ham}(M)
\end{equation}
and has non-trivial higher products
\begin{equation}
\begin{aligned}
 &\mu_1(f):=\dd f~,~~~\mu_1(\beta):=\dd \beta~,~~~
 \mu_2(\alpha_1,\alpha_2):=-\iota_{X_{\alpha_1}}\iota_{X_{\alpha_2}}\varpi~,\\
 &\mu_3(\alpha_1,\alpha_2,\alpha_3):=\iota_{X_{\alpha_1}}\iota_{X_{\alpha_2}}\iota_{X_{\alpha_3}}\varpi~,~~~\mu_4(\alpha_1,\alpha_2,\alpha_3,\alpha_4):=-\iota_{X_{\alpha_1}}\iota_{X_{\alpha_2}}\iota_{X_{\alpha_3}}\iota_{X_{\alpha_4}}\varpi~,\\
 &\mu_2(\alpha,f):=\langle \alpha,\dd f\rangle~,~~~\mu_2(\alpha,\beta):=\dd \langle \alpha,\beta\rangle~,\\
 &\mu_3(\alpha_1,\alpha_2,\beta):=-\mu_2(\alpha_1,\langle \alpha_2,\beta\rangle)+\mu_2(\alpha_2,\langle \alpha_1,\beta\rangle)+\langle \mu_2(\alpha_1,\alpha_2),\beta\rangle~
 \end{aligned}
\end{equation}
for $\alpha,\alpha_i\in \Omega^2_{\rm Ham}(M)$, $\beta,\beta_i\in \Omega^1(M)$ and $f\in \CC^\infty(M)$. This example readily generalizes to $n$-plectic Nambu-Poisson manifolds with Nambu-Poisson tensor $\pi$ of rank $n$.

\subsection{Shlalos and central extensions}

Recall that the ordinary Poisson algebra of a Poisson manifold can be regarded as a central extension of the Lie algebra of Hamiltonian vector fields on a Poisson manifold, where the central subalgebra corresponds to the constant functions. Analogously, the shlalo of an $n$-plectic manifold can be regarded as a central extension of the Lie algebra of Hamiltonian vector fields to a Lie 2-algebra, with the central part given by functions and constant 1-forms \cite{Fiorenza:1304.6292}.

Generalized shlalos of enhanced $n$-plectic manifolds clearly do not fit this picture in a straightforward manner any more. One could simply interpret such shlalos as an extension of the Lie algebra of Hamiltonian vector fields, which is not central. While we do not have a better answer in the general situation, there is a nice picture arising in the shlalo of a $2$-plectic manifold with bracket $\langle-,-\rangle:=\iota_{X_\alpha}\beta-\iota_{X_\beta}\alpha$. The corresponding $L_\infty$-algebra homomorphism was used e.g.\ in \cite{Rogers:2010sc} and \cite{Ritter:2013wpa}.

We extend the Lie algebra of Hamiltonian vector fields to the 2-vector space $\Pi$ given by $\Omega^1_{\rm Ham}(M)\xrightarrow{~0~}\frX_{\rm Ham}(M)$, where we associate a ``Hamiltonian one-form'' $\xi_f:=\dd f$ to each function. On this 2-vector space, we define non-trivial products
\begin{equation}\label{eq:central_brackets}
 \pi_2(X_\alpha,X_\beta)=[X_\alpha,X_\beta]~,~~~\pi_2(X_\alpha,\xi_f)=\dd(\iota_{X_\alpha}\xi_f)~.
\end{equation}
\begin{proposition}
 The 2-vector space $\Pi$ with higher products \eqref{eq:central_brackets} forms a Lie 2-algebra. Moreover, the shlalo of the $2$-plectic manifold with products
\begin{equation}
\begin{aligned}
 \mu_1(f)=\dd f~,~~~\mu_2(\alpha_1,\alpha_2)=-\iota_{X_{\alpha_1}}\iota_{X_{\alpha_2}}\varpi+\dd(\iota_{X_{\alpha_1}}\alpha_2-\iota_{X_{\alpha_2}}\alpha_1)~,~~~
 \mu_2(\alpha,f)=\iota_{X_\alpha}\dd f~,\\
 \mu_3(\alpha_1,\alpha_2,\alpha_3)=-2\iota_{X_{\alpha_1}}\iota_{X_{\alpha_2}}\iota_{X_{\alpha_3}}\varpi-\big(\iota_{[X_{\alpha_2},X_{\alpha_3}]}\alpha_1-\mu_2(\alpha_1,\iota_{X_{\alpha_2}}\alpha_3-\iota_{X_{\alpha_3}}\alpha_2)+\rm{cycl.}\big)\\
 = \iota_{X_{\alpha_1}}\iota_{X_{\alpha_2}}\iota_{X_{\alpha_3}}\varpi \hspace{+9cm}
\end{aligned}
\end{equation}
 is a central extension of $\Pi$.
\end{proposition}
\begin{proof}
 To show that $(\Pi,\pi_2)$ is a Lie 2-algebra, we need to check one associativity condition:
 \begin{equation}
 \begin{aligned}
  \pi_2(X_\alpha,\pi_2(X_\beta,\xi_f))+\pi_2(X_\beta,\pi_2(\xi_f,X_\alpha))+\pi_2(\xi_f,\pi_2(X_\alpha,X_\beta))&=\\
  \dd \iota_{X_\alpha}\dd \iota_{X_\beta}\dd f-\dd \iota_{X_\beta}\dd \iota_{X_{\alpha}}\dd f-\dd \iota_{[X_\alpha,X_\beta]} \dd f&=\\
  \dd \iota_{X_\alpha}\dd \iota_{X_\beta}\dd f-\dd \iota_{X_\beta}\dd \iota_{X_{\alpha}}\dd f-\dd (\CL_{X_\alpha}\iota_{X_\beta}-\iota_{X_\beta}\CL_{X_\alpha})\dd f&=0~.
 \end{aligned}
 \end{equation}
 Also, we have 
 \begin{equation}
  X_{\mu_2(\alpha,\beta)}=\pi_2(X_\alpha,X_\beta)\eand \xi_{\mu_2(\alpha,f)}=\pi_2(X_\alpha,\xi_f)
  \end{equation}
and the maps $X:\Omega^1_{\rm Ham}(M)\rightarrow \frX_{\rm Ham}(M)$ and $\xi:\CC^\infty(M)\rightarrow \Omega^1$ provide a Lie 2-algebra homomorphism with kernel the exact 1-forms and the constant functions. Both are central in $\Pi$.
\end{proof}

This observation also fits from a different perspective: in the case of symplectic manifolds, we centrally extended the Lie-algebra of sections of the corresponding symplectic Lie 1-algebroid $TM$. In the 2-plectic case, we centrally extend the Lie 2-algebra of sections of the corresponding exact symplectic Lie 2-algebroid $TM\oplus T^*M$.

\section{The $L_\infty$-algebra of a symplectic Lie $n$-algebroid}

\subsection{Symplectic Lie $n$-algebroids and associated Lie $n$-algebras}

Lie $n$-algebroids are truncated $L_\infty$-algebroids and they are most conveniently described as general N$Q$-manifolds, cf.\ definition \ref{def:NQ-manifold}. Note that an $\NN$-graded manifold $\CM$ comes with a natural projection onto its {\em body}, which is the manifold $M$ corresponding to the degree $0$-part of $\CM$. Therefore, $\CM$ can be regarded as a fiber bundle $\CE\rightarrow M$ over $M$.
\begin{definition}
 A \uline{symplectic Lie $n$-algebroid} is an N$Q$-manifold endowed with a symplectic structure $\omega$ of degree $n$ satisfying $\CL_Q\omega=0$.
\end{definition}
Strict morphisms of symplectic Lie $n$-algebroids are now readily defined, cf.\ e.g.\ \cite{16298602}.
\begin{definition}\label{def:strict_morphisms}
 A \uline{strict morphism of symplectic Lie $n$-algebroids} $\Gamma:(\CM,\omega,Q)\rightarrow\linebreak (\CM',\omega',Q')$ is a smooth map $\gamma:\CM\rightarrow \CM'$ such that $\omega=\gamma^*\omega'$ and $Q$ and $Q'$ are $\gamma$-related: $\gamma_*(Q_x)=Q'_{\gamma(x)}$ for all $x\in \CM$.
\end{definition}

It is more convenient to encode a symplectic Lie $n$-algebroid in terms of a Hamiltonian and a Poisson bracket.
\begin{proposition}
 A symplectic Lie $n$-algebroid $(\CM,\omega,Q)$ can be equivalently described by a non-degenerate\footnote{i.e.\ the Poisson bivector is an invertible matrix in any coordinates} Poisson bracket $\{-,-\}$ on $\CC^\infty(\CM)$, which is given by the inverse of $\omega$, together with a function $\Theta$ on $\CM$ of degree $n+1$ satisfying $Qf=\{\Theta,f\}$, for all $f\in C^\infty(\CM)$, and $\{\Theta,\Theta\}=0$.
\end{proposition}
\begin{proof}
 It is clear that a non-degenerate Poisson bracket encodes a symplectic form. The fact that a homological vector field of degree 1 with $Q^2=0$ and $\CL_Q\omega=0$ is a Hamiltonian vector field with Hamiltonian $\Theta$ is proved in \cite[Lemma 2.2]{Roytenberg:0203110}. Correspondingly, $Qf=\{\Theta,f\}$, and $Q^2f=0$ then implies $\{\Theta,\Theta\}=0$ via the Jacobi identity 
 \begin{equation}
  \{a,\{b,c\}\}=\{\{a,b\},c\}+(-1)^{(|a|+|\omega|)(|b|+|\omega|)}\{b,\{a,c\}\}
 \end{equation} 
 and the non-degeneracy of the Poisson bracket.
\end{proof}
Recall that given a symplectic form $\omega=\omega_{ij}\dd z^i\wedge \dd z^j$ on a $\RZ$-graded vector space in some local coordinates $z^i$, the inverse matrix $\pi^{ij}=(\omega_{ij})^{-1}$ encodes the Poisson bivector $\pi^{ij}\der{z^i}\wedge \der{z^j}$. We then have
\begin{equation}
 \{f,g\}:=f\overleftarrow{\der{z^i}}\pi^{ij}\overrightarrow{\der{z^j}} g=-(-1)^{(|f|+|\omega|)(|g|+|\omega|)}\{g,f\}~,
\end{equation}
where $\overleftarrow{\der{z^i}}$ denotes a left-acting derivative. Note also that $|\{f,g\}|=|f|+|g|+|\omega|$.

We can now associate a Lie $n$-algebra to each symplectic Lie $n$-algebroid, generalizing \cite[Theorem 4.3]{Roytenberg:1998vn} and \cite[Proposition 8.1]{Zambon:2010ka}.
\begin{theorem}
 Each symplectic Lie $n$-algebroid $(\CM,\{-,-\},\Theta)$ comes with an associated Lie $n$-algebra 
 \begin{equation}
  \sL_{\CM}~:=~\CA_0\rightarrow \CA_1\rightarrow \CA_2\rightarrow \ldots \rightarrow \CA_{n-1}~,
 \end{equation}
 where $\CC^\infty(\CM)=\CA_0\oplus \CA_1\oplus \CA_2\oplus \ldots $ is the decomposition of $\CC^\infty(\CM)$ into parts $\CA_i$ of homogeneous grading $i$. Note that $\CA_0=\CC^\infty(M)$. The higher products read as 
 \begin{subequations}
 \begin{equation}
  \mu_1(e)=\left\{\begin{array}{ll}
  0 & e\in \CA_{n-1}~,\\
  \{\Theta,e_1\} & \mbox{else}~,\\
  \end{array}\right.
 \end{equation}
 and for $k>1$,
 \begin{equation}
  \mu_k(e_1,\ldots,e_k)=\frac{(-1)^{k+k(k+1)/2}}{(k-1)!}B_{k-1}\sum_\sigma (-1)^\sigma\{\ldots\{\delta e_{\sigma(1)},e_{\sigma(2)}\}\ldots,e_{\sigma(k)}\}~.
 \end{equation}
 \end{subequations}
 Here, $B_k$ are the (first) Bernoulli numbers\footnote{i.e.\ $B_{0},\ldots=1,-\tfrac12,\tfrac{1}{6},0,-\tfrac{1}{30},0,\tfrac{1}{42},\ldots$}, $e,e_1,\ldots,e_k\in \sL_{\CM}$, $\sigma$ runs over all permutations, $(-1)^\sigma$ is the Koszul sign of the permutation and 
 \begin{equation}
  \delta e=\left\{\begin{array}{ll}
  \{\Theta,e\} & e\in \CA_{n-1}~,\\
  0 & \mbox{else}~.\\
  \end{array}\right.
 \end{equation}
\end{theorem}
\begin{proof}
 One readily checks the homotopy Jacobi identities for the $\mu_k$ for
 low values of $k$. In its full generality, this theorem is a direct
 consequence of \cite[Theorem 3]{Getzler:1010.5859}, which in turn can
 be regarded as a special case of \cite[Theorem
 5.5]{Fiorenza:0601312}. The additional sign $(-1)^{k(k+1)/2}$ in the
 definition of $\mu_k$ compared to \cite{Getzler:1010.5859}
is due to different sign conventions in the definition of our homotopy Jacobi identities.
\end{proof}

As an example, we give the details for a symplectic Lie 3-algebroid. 
\begin{example}
 We start from the N$Q$-manifold $T^*[3]V[1]$, where $V$ is some vector bundle of rank $r$ over a manifold $M$ of dimension $d$. The total space decomposes as $\CE=E_1\oplus E_2\oplus E_3$ with ranks $r$, $r$, $d$. On the base and the fibers, we introduce coordinates $x^i,\xi^\alpha,\zeta_\alpha,p_i$ of degrees $0$, $1$, $2$ and $3$, for which we have the natural symplectic form
 \begin{equation}
  \omega=\dd p_i\wedge \dd x^i+\dd \xi^\alpha\wedge \dd \zeta_\alpha~.
 \end{equation}
 One readily shows that the Hamiltonian, a function of degree 4, is necessarily of the form
 \begin{equation}
  \Theta=\rho^i_\alpha \xi^\alpha p_i+\tfrac12 r_{\alpha\beta}^\gamma \xi^\alpha\xi^\beta\zeta_\gamma+\tfrac12 s^{\alpha\beta}\zeta_\alpha\zeta_\beta+\tfrac{1}{4!}t_{\alpha\beta\gamma\delta}\xi^\alpha\xi^\beta\xi^\gamma\xi^\delta~,
 \end{equation}
 cf.\ e.g.\ \cite{Kotov:2010wr}, with further constraints on the coefficients $\rho^i_\alpha$, $r_{\alpha\beta}^\gamma$, $s^{\alpha\beta}$ and $t_{\alpha\beta\gamma\delta}$. The underlying complex $\CA_0\oplus\CA_1\oplus \CA_2$ can be identified with
 \begin{equation}
  \sL_{\CM}~\cong~\CC^\infty(M)~\oplus~V^*~\oplus~(V\oplus \wedge^2 V^*)~.
 \end{equation}
 The higher products are then readily calculated. For example, we have
 \begin{equation}
 \begin{aligned}
  \mu_1(f)&=\{\Theta,f\}=\rho^i_\alpha \xi^\alpha \dpar_i f~,\\
  \mu_1(f_\alpha\xi^\alpha)&=\rho^i_\beta \xi^\beta(\dpar_i f_\alpha)\xi^\alpha+\tfrac12 r^\gamma_{\alpha\beta}\xi^\alpha\xi^\beta f_\gamma+s^{\alpha\beta}\zeta_\alpha f_\beta~,
 \end{aligned}
 \end{equation}
 where $\dpar_i=\der{x^i}$ and $f,f_\alpha,\ldots\in \CC^\infty(M)$.
\end{example}

\subsection{Vinogradov algebroids and shlalos}

Before continuing with automorphisms of symplectic Lie $n$-algebroids,
let us consider the relation between shlalos and symplectic Lie
$n$-algebroids in more detail, further motivating our
later discussion.

An $n$-plectic manifold $(M,\varpi)$ comes with the following natural symplectic Lie $n$-algebroid, see e.g.\ \cite{Bi:1003.1350}:
\begin{definition}\label{Vinogradov definition}
 The \uline{Vinogradov algebroid}\,\footnote{The underlying bracket was first introduced in \cite{MR1074539}. A more general definition of Vinogradov algebroids can be found in \cite{Gruetzmann:2014ica}.} $\CV_n$ of a manifold $M$, is given by the vector bundle $TM\oplus \wedge^{n-1} T^*M$ with anchor map, bracket and non-degenerate bilinear form
 \begin{equation}
  \begin{aligned}
    \rho(X_1+\alpha_1)&:=X_1~,\\
    \lsb X_1+\alpha_1,X_2+\alpha_2\rsb&:=[X_1,X_2]+\CL_{X_1}\alpha_2-\CL_{X_2}\alpha_1+\tfrac12(\dd\iota_{X_2}\alpha_1-\dd\iota_{X_1}\alpha_2)~,\\
    \big(X_1+\alpha_1,X_2+\alpha_2\big)&:=\tfrac12\big(\iota_{X_1}\alpha_2+\iota_{X_2}\alpha_1\big)~,
  \end{aligned}
 \end{equation}
 where $X_{1,2}\in \Gamma(TM)$ and $\alpha_{1,2}\in \Gamma(\wedge^{n-1} T^*M)$. As in the case of Courant algebroids, we also define
 \begin{equation}
  T(e_1,e_2,e_3):=-\tfrac13\Big(\big(\lsb e_1,e_2\rsb,e_3\big)+\big(\lsb e_2,e_3\rsb,e_1\big)+\big(\lsb e_3,e_1\rsb,e_2\big)\Big)~.
 \end{equation}
 The \uline{twisted Vinogradov algebroid} of an $n$-plectic manifold $(M,\varpi)$ is the same as the Vinogradov algebroid $\CV_n$ with bracket replaced by 
 \begin{equation}
    \lsb X_1+\alpha_1,X_2+\alpha_2\rsb_{\varpi}:=[X_1,X_2]+\CL_{X_1}\alpha_2-\CL_{X_2}\alpha_1+\tfrac12(\dd\iota_{X_2}\alpha_1-\dd\iota_{X_1}\alpha_2)+\iota_{X_1}\iota_{X_2}\varpi~.
 \end{equation}
\end{definition}
In the case $n=2$, the above Vinogradov algebroid is an exact Courant algebroid. In this case, the twist element $\varpi(X_1,X_2,X_3)=\big(X_1,\lsb X_2,X_3\rsb_\varpi\big)$ yields a 3-form representing the {\em \v Severa class} which classifies Courant algebroid structures on $TM\oplus T^*M$ \cite{Severa:1998aa}. The twist element $\varpi$ above generalizes this to Vinogradov algebroids.

The properties of the structure maps on Vinogradov algebroids were studied e.g.\ in \cite{Gualtieri:2003dx,Bi:1003.1350,Gruetzmann:2014ica}.
\begin{proposition}\label{Vinogradov properties}
 Let $e_1,e_2,e_3$ be sections of the vector bundle $E:=TM\oplus \wedge^{n-1} T^*M$ underlying the Vinogradov algebroid and let $f\in \CC^\infty(M)$. Then
 \begin{equation}
  \begin{aligned}
    \lsb e_1,\lsb e_2,e_3\rsb\rsb+\lsb e_2,\lsb e_3,e_1\rsb\rsb+\lsb e_3,\lsb e_1,e_2\rsb\rsb&=-\dd T(e_1,e_2,e_3)~,\\
    \lsb e_1,f e_2\rsb&=f\lsb e_1,e_2\rsb+\rho(e_1)fe_2-\dd f\wedge (e_1,e_2)~,\\
    \rho(\lsb e_1,e_2\rsb)&=[\rho(e_1),\rho(e_2)]~.
  \end{aligned}
 \end{equation}
\end{proposition}
\begin{proof} 
For the untwisted case, the above is just the well known exact Courant algebroid
over $TM\oplus \wedge^{n-1} T^*M$, as described in \cite{Gruetzmann:2014ica}. For the twisted Vinogradov algebroid, one can calculate
  that
  \begin{equation}
  T_\varpi(e_1,e_2,e_3)=T(e_1,e_2,e_3)-\tfrac{3}{2}\varpi(X_1,X_2,X_3)~,  
  \end{equation}
while the elements of the Jacobiator become
\begin{align}
  \lsb e_1,\lsb e_2,e_3\rsb_\varpi\rsb_\varpi=&\lsb e_1\lsb
                                                e_2,e_3\rsb\rsb+\mathcal{L}_{X_1}\iota_{X_2}\iota_{X_3}\varpi
                                                -\tfrac{1}{2}\dd\iota_{X_1}\iota_{X_2}\iota_{X_3}\varpi+
\iota_{X_1}\iota_{[X_2,X_3]}\varpi~\nonumber\\
\label{Jac} =&\lsb e_1\lsb
                                                e_2,e_3\rsb\rsb+\iota_{X_1}\dd\iota_{X_2}\iota_{X_3}\varpi
                                               +\tfrac{1}{2}\dd\iota_{X_1}\iota_{X_2}\iota_{X_3}\varpi+
\iota_{X_1}\iota_{[X_2,X_3]}\varpi~.
\end{align}
Making use of the fact that
\begin{equation*}
  \iota_{X_1}\iota_{X_2}\iota_{X_3}\dd\varpi=\left(\iota_{X_1}\iota_{[X_2,X_3]}
    + \iota_{X_1}\dd
    \iota_{X_2}\iota_{X_3}+\text{cycl.}(X_1,X_2,X_3)\right)\varpi +2\dd \iota_{X_1}\iota_{X_2}\iota_{X_3}\varpi~,
\end{equation*}
one can readily calculate from (\ref{Jac}) that
\begin{equation}
  \label{eq:2}
  \lsb e_1,\lsb e_2,e_3\rsb_\varpi\rsb_\varpi+\text{cycl.}=(\lsb e_1\lsb
                                                e_2,e_3\rsb\rsb+\text{cycl.}) +\tfrac{3}{2}\dd\iota_{X_1}\iota_{X_2}\iota_{X_3}\varpi-\tfrac{1}{2}\iota_{X_1}\iota_{X_2}\iota_{X_3}\dd\varpi~.
\end{equation}
We therefore have for the twisted Vinogradov algebroid that the
Jacobiator equation becomes
\begin{equation}
  \label{eq:3}
    \lsb e_1,\lsb e_2,e_3\rsb_\varpi\rsb_\varpi+\text{cycl.}=-\dd T_\varpi(e_1,e_2,e_3)- \tfrac{1}{2} \iota_{X_1}\iota_{X_2}\iota_{X_3}\dd\varpi~,
\end{equation}
so that we require $\varpi$ to be closed for proposition
\ref{Vinogradov properties} to still hold in the twisted case. Note
that neither the degree of the form $\varpi$, nor of the elements $\alpha_i$,
have been used in these calculations. That is, what held for
Courant and twisted Courant algebroids, cf.\ \cite{Gualtieri:2003dx}, still goes through in the same
way.

To check the second equality in \ref{Vinogradov properties} for both
the twisted as untwisted case, just use $f\cdot
e=f\cdot (X+\alpha)=fX+f\alpha$. The last equality in the
proposition trivially follows from definition \ref{Vinogradov definition}.
 \end{proof}
The Vinogradov algebroid is readily interpreted as a symplectic Lie $n$-algebroid, cf.\ e.g.\ \cite{Gruetzmann:2014ica}. 
\begin{proposition}
 Consider the N$Q$-manifold $\CM=T^*[n]T[1]M$ with fiber coordinates $x^i,\xi^i,\zeta_i,p_i$ of grading $0$, $1$, $n-1$ and $n$. Together with 
 \begin{equation}
  \omega=\dd p_i\wedge \dd x^i+\dd \xi^i\wedge \dd \zeta_i\eand \Theta=-\xi^ip_i~,
 \end{equation}
 $\CM$ becomes a symplectic Lie $n$-algebroid which contains the Vinogradov algebroid if we identify $\Gamma(E)\equiv\CA_{n-1}$,
 \begin{equation}
  \lsb e_1,e_2\rsb\equiv\tfrac12\left(\{\{\Theta,e_1\},e_2\}-\{\{\Theta,e_2\},e_1\}\right)\eand (e_1,e_2)\equiv\tfrac12\{e_1,e_2\}~.
 \end{equation}
 The twisted Vinogradov algebroid with $n$-plectic structure 
 \begin{equation}
  \varpi=\tfrac{1}{(n+1)!}\varpi_{i_1\ldots i_{n+1}}\dd x^{i_1}\wedge \ldots \wedge \dd x^{i_{n+1}}~,
 \end{equation}
 is contained in the symplectic Lie $n$-algebroid given by $\CM$ with
 \begin{equation}
  \omega=\dd p_i\wedge \dd x^i+\dd \xi^i\wedge \dd \zeta_i\eand \Theta=-\xi^ip_i+\tfrac{1}{(n+1)!}\varpi_{i_1\ldots i_{n+1}}\xi^{i_1}\ldots \xi^{i_{n+1}}~.
 \end{equation}
\end{proposition}
\begin{proof}
 An element $a$ of $\CA_{n-1}$ is of the form
 \begin{equation}
  a=\tfrac{1}{(n-1)!}\alpha_{i_1\ldots i_{n-1}}\xi^{i_1}\ldots\xi^{i_{n-1}}+X^i\zeta_i
 \end{equation}
  and thus amounts indeed to a section of $TM\oplus \wedge^{n-1}T^*M$. The remaining proof amounts to straightforward computations. First, we note that 
 \begin{equation}
 \begin{aligned}
  2(a,b)&=\{a,b\}=a\overleftarrow{\der{x^i}}\overrightarrow{\der{p_i}}b-a\overleftarrow{\der{p_i}}\overrightarrow{\der{x^i}}b+a\overleftarrow{\der{\zeta_i}}\overrightarrow{\der{\xi^i}}b-(-1)^{n-1}a\overleftarrow{\der{\xi^i}}\overrightarrow{\der{\zeta_i}}b\\
  &=\tfrac{1}{(n-2)!}X^{i}\beta_{ii_2\ldots i_{n-1}}\xi^{i_2}\ldots \xi^{i_{n-1}}-(-1)^{n-1}(-1)^{n-2}\tfrac{1}{(n-2)!}\alpha_{ii_2\ldots i_{n-1}}\xi^{i_2}\ldots \xi^{i_{n-1}}Y^{i}\\
  &=\iota_{X}\beta+\iota_{Y}\alpha~,
 \end{aligned}
 \end{equation}
 where $X+\alpha$ and $Y+\beta$ are identified with $a$ and $b$, respectively, via $\xi^i\leftrightarrow \dd x^i$ and $\zeta_i\leftrightarrow \der{x^i}$. Also,
 \begin{equation*}
 \begin{aligned}
  \{\{\Theta,a\},b\}=&X^i\dpar_i b-Y^i\dpar_i a+\xi^i(\dpar_i X^j)\tfrac{1}{(n-2)!} \beta_{ji_1\ldots i_{n-2}}\xi^{i_1}\ldots \xi^{i_{n-2}}+\\
  &+\xi^j Y^i \partial_j \tfrac{1}{(n-2)!}\alpha_{il_1\ldots l_{n-2}}\xi^{l_1}\cdots\xi^{l_{n-2}}+\tfrac{1}{(n-1)!}X^{i}Y^{j}\varpi_{ijl_1\ldots l_{n-1}}\xi^{l_1}\ldots \xi^{l_{n-1}}\\
  =&[X,Y]+\CL_X\beta+\iota_Y\dd \alpha+\iota_X\iota_Y\varpi~,
 \end{aligned}
 \end{equation*}
 where we again identified $X+\alpha$ and $Y+\beta$ with $a$ and $b$ as above. Upon antisymmetrization, we get the correct expression for $\lsb a,b\rsb=\tfrac12(\{\{\Theta,a\},b\}-\{\{\Theta,b\},a\})$.
\end{proof}
If $\varpi$ is integral, then the case $n=1$ corresponds to the Atiyah algebroid for the principal $\sU(1)$-bundle with first Chern class $\varpi$. Similarly, the case $n=2$ corresponds to a higher Atiyah algebroid of a $\sU(1)$-bundle gerbe with Dixmier-Douady class $\varpi$.

\begin{remark}\label{Remark Vinogradov}
 In general, we have the following equations for the twisted
 Vinogradov algebroid:
 \begin{equation}
 \begin{aligned}
  &\{\{\Theta,\tfrac{1}{p!}\alpha_{i_1\ldots i_p}\xi^{i_1}\ldots \xi^{i_p}\},\tfrac{1}{q!}\beta_{j_1\ldots j_q}\xi^{j_1}\ldots \xi^{j_q}\}=0~,\\
  &\{\{\Theta,\tfrac{1}{p!}\alpha_{i_1\ldots i_p}\xi^{i_1}\ldots \xi^{i_p}\},Y^i\zeta_i\}=(-1)^{n+p} Y^j(\dpar_j\tfrac{1}{p!}\alpha_{i_1\ldots i_p}-\dpar_{i_1}\tfrac{1}{(p-1)!}\alpha_{ji_2\ldots i_p})\xi^{i_1}\ldots \xi^{i_p}~,\\
  &\{\{\Theta,X^i\zeta_i\},\tfrac{1}{p!}\beta_{i_1\ldots i_p}\xi^{i_1}\ldots \xi^{i_p}\}=(X^j\dpar_j\tfrac{1}{p!}\beta_{i_1\ldots i_p}+(\dpar_{i_1}X^j)\tfrac{1}{(p-1)!}\beta_{ji_2\ldots i_p})\xi^{i_1}\ldots \xi^{i_p}~,\\
  &\{\{\Theta,X^i\zeta_i\},Y^j\zeta_j\}=(X^j\dpar_jY^k-Y^j\dpar_j X^k)\zeta_k-\tfrac{1}{(n-1)!}X^jY^k\varpi_{jki_1\ldots i_{n-1}}\xi^{i_1}\ldots \xi^{i_{n-1}}~.
 \end{aligned}
 \end{equation}
In a more compact form, where $\alpha,\,\beta$ are now generic $p$- and $q$-forms, these can be rewritten as
\begin{equation}
 \begin{aligned}
  \{\{\Theta,\alpha\},\beta\}&=0~,~~~&
\{\{\Theta,\alpha\},Y\}&=(-1)^{n+p}\iota_Y\dd\alpha~,\\
\{\{\Theta,X\},\beta\}&=\mathcal{L}_X\beta~,~~~&
\{\{\Theta,X\},Y\}&=[X,Y]+\iota_X\iota_Y\varpi~.
 \end{aligned}
\end{equation}
\end{remark}

It was shown in \cite{Rogers:2010nw,Rogers:2011zc} that a Lie 2-algebra equivalent to the shlalo of a 2-plectic manifold is contained in the Lie 2-algebra associated to the Vinogradov algebroid of the 2-plectic manifold. We are now ready to generalize this statement to all $n$.
\begin{theorem}\label{thm:embedding_twisted}
 The associated Lie $n$-algebra of the twisted Vinogradov algebroid of an $n$-plectic manifold $(M,\varpi)$ contains a sub Lie $n$-algebra which is equivalent to the shlalo of $(M,\varpi)$. This sub Lie $n$-algebra has underlying complex
 \begin{equation}
  \CC^\infty(M)\xrightarrow{~\xi^i\dpar_i~}\CA_1(M)\xrightarrow{~\xi^i\dpar_i~}\ldots\xrightarrow{~\xi^i\dpar_i~}\CA_{n-1}^{\rm Ham}(M)~,
 \end{equation}
 where we identify $\CA_i(M)\cong \Omega^i(M)$ for $i\leq {n-1}$ and $\CA^{\rm Ham}_{n-1}(M)\cong \frX_{\rm Ham}(M)\oplus \Omega^{n-1}_{\rm Ham}$. Here we embed elements $\alpha\in\Omega^{n-1}_{\rm Ham}(M)$ into $\frX_{\rm Ham}(M)\oplus \Omega_{\rm Ham}^{n-1}(M)$ as $X_\alpha+\alpha$, where $X_\alpha$ is the Hamiltonian vector field of $\alpha$ with respect to $\varpi$.
\end{theorem}
\begin{proof}
We denote the Lie $n$-algebra of the Vinogradov algebroid by $\sL_{\CV_n}$ and the shlalo of $(M,\varpi)$ by $\sL_M$. Their higher products are denoted by $\mu_i(-,\ldots,-)$ and $\pi_i(-,\ldots,-)$, respectively. We readily compute that the lowest two products of $\sL_{\CV_n}$ read as 
\begin{equation}
\begin{aligned}
  &\mu_1(\alpha+\gamma+X)=\dd \gamma~,\\
  &\mu_2(\alpha_1+\gamma_1+X_1,\alpha_2+\gamma_2+X_2)=[X_1,X_2]+\CL_{X_1}\alpha_2-\CL_{X_2}\alpha_1+\\&\hspace{2.5cm}+\tfrac12\left(\CL_{X_1}\gamma_2-(-1)^{n-1+|\gamma_1|}\CL_{X_2}\gamma_1+\dd \iota_{X_2}\alpha_1-\dd \iota_{X_1}\alpha_2\right)+\iota_X\iota_Y\varpi~,
\end{aligned}
\end{equation}
where $X,X_{1,2}\in \frX(M)$, $\alpha,\alpha_{1,2}\in \Omega^{n-1}(M)$ and $\gamma,\gamma_{1,2}\in \oplus_{i=0}^{n-2}\Omega^i(M)$. Applying an isomorphism of $L_\infty$-algebras with $\Psi^1_1=\id$, 
\begin{equation}
\Psi^1_2(\beta_1+X_1,\beta_2+X_2)=\tfrac12\left((-1)^{n-1+|\beta_2|}\iota_{X_1}\beta_2-(-1)^{n-1+|\beta_1|}\iota_{X_2}\beta_1\right)~,
\end{equation}
where $\beta_{1,2}\in \oplus_{i=0}^n\Omega^i(M)$ and $\Psi^1_i=0$ for $i>2$, yields the new products
\begin{equation}\label{eq:4.30}
\begin{aligned}
  \mu'_1(\alpha+\gamma+X)&=\dd \gamma~,\\
  \mu'_2(\alpha_1+\gamma_1+X_1,\alpha_2+\gamma_2+X_2)&=[X_1,X_2]+\iota_{X_1}\dd \alpha_2-\iota_{X_2}\dd \alpha_1+\iota_{X_1}\iota_{X_2}\varpi~.
\end{aligned}
\end{equation}
We now restrict to the above mentioned sub Lie $n$-algebra. That is, we consider the subset of ``diagonal sections'' $\CA_{n-1}^{\rm Ham}\subset \frX_{\rm Ham}(M)\oplus \Omega_{\rm Ham}^{n-1}(M)$ of the form $X_\alpha+\alpha$, where $X_\alpha$ is the Hamiltonian vector field of $\alpha$. We denote the resulting sub Lie $n$-algebra by $\sL^{(0)}_{\CV_n}$ and its higher products by $\mu^{(0)}_k(-,\ldots,-)$. Clearly, we have
\begin{equation}\label{eq:matching-2}
 \mu^{(0)}_1(e)|_{\Omega^\bullet}=\pi_1(e|_{\Omega^\bullet})\eand \mu^{(0)}_2(e_1,e_2)|_{\Omega^\bullet}=\pi_2(e_1|_{\Omega^\bullet},e_1|_{\Omega^\bullet})~,
\end{equation}
where $e,e_{1,2}\in \sL^{(0)}_{\CV_n}$ and $e|_{\Omega^\bullet}$ denotes the restriction of $e$ to its components in $\Omega^\bullet$. Since we consider diagonal sections, this restriction is an isomorphism $\CA_{n-1}^{\rm Ham}\cong \Omega_{\rm Ham}^{n-1}(M)$. To show that the sub Lie $n$-algebra $\sL^{(0)}_{\CV_n}$ is isomorphic to $\sL_M$, it remains to find a sequence of isomorphisms of $L_\infty$-algebras which, when applied to $\sL^{(0)}_{\CV_n}$, leads to relations for all higher brackets $\mu^{(0)}_k(-,\ldots,-)$ and $\pi_k(-,\ldots,-)$ analogous to \eqref{eq:matching-2}.

We start by considering\footnote{This case is already covered in the theorem by Roytenberg and Weinstein \cite{Roytenberg:1998vn}. Moreover, we could deal with it by direct computation. However, we wish to start our iterative proof already at this point, presenting all necessary arguments.} $\mu^{(0)}_3(e_1,e_2,e_3)$ with $e_{1,2,3}\in \sL^{(0)}_{\CV_n}$ and $d:=|e_1|+|e_2|+|e_3|=0$. Due to \eqref{eq:matching-2} and the higher homotopy relation $\mu_1\circ \mu_3=\mu_2\circ \mu_2+\mu_3\circ \mu_1$, which restricts to $\mu_1\circ \mu_3=\mu_2\circ \mu_2$ for $d=0$, we know that there is a map $\rho_3^{(0)}$ of degree $-1+d$ such that
\begin{equation}
 \left.\left(\mu^{(0)}_3(e_1,e_2,e_3)-\rho^{(0)}_3(e_1,e_2,e_3)\right)\right|_{\Omega^\bullet}=\pi_3(e_1|_{\Omega^\bullet},e_2|_{\Omega^\bullet},e_3|_{\Omega^\bullet})~,~~~\dd \rho^{(0)}_3(e_1,e_2,e_3)=0~.
\end{equation}
That is, the lower products fix the higher ones up to a closed term. We now distinguish two cases. First, for $n=2+d$, $\rho^{(0)}_3(e_1,e_2,e_3)\in \Omega^0(M)$ and a closed function is necessarily a constant one. Note that $\rho^{(0)}_3(e_1,e_2,e_3)$ is a sum of terms constructed from exterior derivatives and contractions with arbitrary vector fields acting on arbitrary forms and a given multisymplectic $(n+1)$-form, and we will call such terms form-contraction terms. However, the only form-contraction term which is constant is the vanishing one. We conclude that in the case $n=2+d$, $\rho^{(0)}_3(e_1,e_2,e_3)=0$ and we do not need to apply any further isomorphism. We can therefore continue to work with $\sL^{(1)}_{\CV_n}:=\sL^{(0)}_{\CV_n}$. In the case $n>2+d$, we observe that the only form-contraction terms which are closed are necessarily exact. We can therefore find a map $\psi_3(e_1,e_2,e_3)\mapsto e\in \sL^{(0)}_{\CV}$ of degree $-2$ with $\dd \psi_3(e_1,e_2,e_3)=\rho^{(0)}_3(e_1,e_2,e_3)$. Applying the isomorphism of $L_\infty$-algebras with
\begin{equation}
 \Psi^1_1=\id~,~~~\Psi^1_3=\psi_3\eand\Psi^1_i=0~~~\mbox{otherwise}
\end{equation}
yields an $L_\infty$-algebra $\sL^{(1)}_{\CV_n}$ with higher products $\mu^{(1)}_k(-,\ldots,-)$ such that
\begin{equation}
\mu^{(1)}_3(e_1,e_2,e_3)|_{\Omega^\bullet}=\pi_3(e_1|_{\Omega^\bullet},e_2|_{\Omega^\bullet},e_3|_{\Omega^\bullet})~.
\end{equation}
Since the higher products $\mu^{(1)}_k(-,\ldots,-)$ with $k\leq 2$ remain unaffected by the isomorphism, we retain \eqref{eq:matching-2}. Altogether we constructed an isomorphism which made one more bracket in $\sL^{(1)}_{\CV_n}$ with arguments of a certain total degree agree with $\sL_M$.

We readily apply this procedure $n-2$ more times to $\mu^{(1)}_3(e_1,e_2,e_3)$ for all $|e_1|+|e_2|+|e_3|=1,2,3,\ldots,n-2$, which leaves us with an $L_\infty$-algebra $\sL^{(n-1)}_{\CV_n}$ such that 
\begin{equation}
 \begin{aligned}
  \mu^{(n-1)}_1(e)|_{\Omega^\bullet}&=\pi_1(e|_{\Omega^\bullet})~,\\
  \mu^{(n-1)}_2(e_1,e_2)|_{\Omega^\bullet}&=\pi_2(e_1|_{\Omega^\bullet},e_1|_{\Omega^\bullet})~,\\
  \mu^{(n-1)}_3(e_1,e_2,e_3)|_{\Omega^\bullet}&=\pi_3(e_1|_{\Omega^\bullet},e_2|_{\Omega^\bullet},e_3|_{\Omega^\bullet})
 \end{aligned}
\end{equation}
for all $e,e_{1,2,3}\in \sL^{(n-1)}_{\CV_n}$.

We then continue to apply this procedure to $\mu^{(n-1)}_4(e_1,e_2,e_3,e_4)$, starting from $|e_1|+|e_2|+|e_3|+|e_4|=0$. This works because the homotopy Jacobi identity $\mu_1\circ \mu_4=\mu_2\circ \mu_3+\mu_4\circ \mu_1$, where the sum of the degrees of the arguments of the second $\mu_4$ is by one lower than that of the first $\mu_4$, fixes $\mu_4$ again up to closed terms. Further iterations then deal with all higher orders.

Since $\sL^{(0)}_{\CV_n}$ is concentrated in degrees $1-n,\ldots,0$, there can only be finitely many non-vanishing products. Therefore our iterative algorithm indeed terminates at some point. We then have constructed an $\sL^{(q)}_{\CV_n}\cong \sL^{(0)}_{\CV_n}$ for some $q\in \NN$ such that $\sL^{(q)}_{\CV_n}\cong \left.\sL^{(q)}_{\CV_n}\right|_{\Omega^\bullet}=\sL_M$.
 \end{proof}

\subsection{Automorphisms of associated Lie $n$-algebras}

Similarly to our discussion in section \ref{sec:shlalos}, we can now ask if the assignment of a Lie $n$-algebra to a symplectic Lie $n$-algebroid covers weak automorphisms of Lie $n$-algebras. That is, given an associated Lie $n$-algebra of a symplectic Lie $n$-algebroid, will an automorphism of Lie $n$-algebras yield the associated Lie $n$-algebra of an automorphic Lie $n$-algebroid? The answer is no, in general, as is easily seen from the following example.
\begin{example}
  Consider the Lie 2-algebroid $T[1]M$ with canonical symplectic structure and $\Theta=0$. Each strict automorphism of Lie 2-algebroids will lead to a Lie 2-algebroid with $\Theta=0$.  The associated Lie 2-algebra of $T[1]M$ has trivial higher products $\mu_i=0$, and the same is true for all strictly automorphic Lie 2-algebroids. A weak Lie 2-algebra morphism, however, can generate non-trivial higher products, cf.\ e.g.\ \eqref{eq:example-extended-shlalo-2-plectic}.
\end{example}
We thus have to turn to more general morphisms of symplectic Lie $n$-algebroids than the morphisms of N$Q$-manifolds of definition \ref{def:strict_morphisms}.

First, note that morphisms of Courant algebroids were discussed before in \cite{Bursztyn:0801.1663} and \cite{LiBland:0811.4554} using Dirac structures. This notion of morphism essentially amounts again to morphisms of N$Q$-manifolds and therefore is too strict for our purposes.

To allow for morphisms of N$Q$-manifolds which go beyond the strict ones, we clearly need to endow the N$Q$-manifolds with additional structure. It is therefore natural to follow the same route as in the case of multisymplectic manifolds. Indeed, lifting an N$Q$-manifold to an enhanced manifold is fully sufficient. 
\begin{definition}
 An \uline{$n$-enhanced N$Q$-manifold} is an N$Q$-manifold endowed with totally (graded) antisymmetric $k$-ary multilinear maps, $k\leq n$,
 \begin{equation}
  \langle -,\ldots, -\rangle_k:\CC^\infty(\CM)^{\wedge k}\rightarrow \CC^\infty(\CM)
 \end{equation}
 of degree $1-k$.
\end{definition}
Just as in the case of enhanced multisymplectic manifolds we can use the brackets to define a morphism of $L_\infty$-algebras acting on the trivial $L_\infty$-algebra
\begin{equation}
 \sL_{\CM}^0:= \CC^\infty(M)\xrightarrow{~Q~}\CA_1(M)\xrightarrow{~Q~}\CA_2(M) \xrightarrow{~Q~}\ldots
\end{equation}
by putting $(\Phi_M)^1_1=\id$ and $(\Phi_M)^1_k=\langle -,\ldots,-\rangle_k$ for $k>1$. 

Automorphisms of enhanced N$Q$-manifolds are now readily defined.
\begin{definition}
 An \uline{automorphism on an enhanced N$Q$-manifold} given in terms of data $(\CM,\Theta,\{-,-\},\langle -,\ldots,-\rangle)$ is an automorphism $\Psi^1$ acting on the $L_\infty$-algebra $\sL_{\CM}^0$ such that $\Psi^1_1$ is an N$Q$-map $\phi$. The transformed brackets $\langle -,\ldots, -\rangle'$ are given by 
 \begin{equation}
  \langle-,\ldots,-\rangle'_i=(\Psi\circ \Phi_M)_i\circ(\phi^*\wedge\ldots\wedge \phi^*)~.
 \end{equation}
\end{definition}
Finally, we can now generalize the association of an $L_\infty$-algebra to an enhanced symplectic Lie $n$-algebroid.
\begin{definition}
 The \uline{associated $L_\infty$-algebra of an enhanced symplectic Lie $n$-algebroid} is the result of applying the weak morphism $\Phi_M$ of $L_\infty$-algebras encoded in the brackets $\langle-,\ldots,-\rangle$ to the associated $L_\infty$-algebra of the symplectic Lie $n$-algebroid (without brackets).
\end{definition}
At this point, we have reached our last goal. We found an extension of symplectic Lie $n$-algebroids such that automorphisms of associated Lie $n$-algebras yield associated Lie $n$-algebras of automorphic symplectic Lie $n$-algebroids. We assume this feature to be quite useful. In particular, we can now directly show the connection between the shlalo of an $n$-plectic manifold and the associated Lie $n$-algebra of the enhanced Vinogradov algebroid $\CV_n$:
\begin{proposition}
 Consider the enhanced twisted Vinogradov algebroid $\CV_n$ endowed with brackets $\langle -,-\rangle$ corresponding to the concatenation of isomorphisms constructed in theorem \ref{thm:embedding_twisted}. Then the associated $L_\infty$-algebra can be restricted in degree $0$ to elements of $\frX_{\rm Ham}(M)\oplus\Omega^{n-1}_{\rm Ham}(M)$ of the form $X_\alpha+\alpha$, where $X_\alpha$ is the Hamiltonian vector field of $\alpha$ with respect to the $n$-plectic form $\varpi$. The resulting sub Lie $n$-algebra is identical to the shlalo of an $n$-plectic manifold $(M,\varpi)$.
\end{proposition}

\section*{Acknowledgements}
We would like to thank Thomas Strobl for discussions related to $Q$-manifolds. We are also grateful to Brano Jur\v co for helpful conversations. The work of CS was partially supported by the Consolidated Grant ST/L000334/1 from the UK Science and Technology Facilities
Council.

\appendices

\subsection{Differential graded coalgebras and their morphisms}\label{app:Definitions}

In this appendix, we collect a number of definitions related to the coalgebra description of $L_\infty$-algebras, fixing our notation. A good reference for these definitions is \cite{0387950680}.
\begin{definition}
A (coassociative) \uline{graded coalgebra} is a graded vector space $C$ endowed with two linear maps of degree 0, the \uline{coproduct} $\Delta:C\rightarrow C\otimes C$ and the \uline{counit} $\eps:C\rightarrow \FC$, which satisfy the conditions
\begin{equation}
\begin{aligned}
 (\Delta\otimes \id)\circ \Delta=(\id\otimes \Delta)\circ \Delta~~&:~~C\rightarrow C\otimes C\otimes C~,\\
 (\id\otimes \eps)\circ\Delta=(\eps \otimes \id)\circ \Delta=\id_C~~&:~~C\rightarrow C~. 
\end{aligned}
\end{equation}
A graded coalgebra is called \uline{cocommutative}, if $\tau\circ \Delta=\Delta$, where $\tau$ is the twist operator $\tau(x\otimes y)=(-1)^{|x|\,|y|}y\otimes x$ for $x,y\in C$.
\end{definition}
Note that we adopt the usual Koszul sign convention for operators acting on elements of a graded vector space $C$. That is, $(f\otimes g)(x\otimes y):=(-1)^{|g|\,|x|}f(x)\otimes g(y)$ for $x,y\in C$ and $f,g\in \sEnd(C)$, where $|g|$ and $|x|$ denote the degrees of $g$ and $x$, respectively.
\begin{definition}
 A \uline{coalgebra morphism} from the coalgebra $(C,\Delta)$ to the coalgebra $(C',\Delta')$ is a linear map $\Psi:C\rightarrow C'$ of degree 0 such that $(\Psi\otimes \Psi)\circ \Delta=\Delta'\circ \Psi$.
\end{definition}

\begin{definition}
 A \uline{coderivation} on a graded coalgebra $(C,\Delta)$ is a linear map $\CD:C\rightarrow C$ which satisfies
 \begin{equation}\label{eq:co-Leibniz}
  \Delta\circ \CD=(\CD\otimes \id+\id\otimes \CD)\circ \Delta\eand \eps\circ \CD=0~.
 \end{equation}
\end{definition}

\begin{definition}
 A (coassociative) \uline{differential graded coalgebra} is a graded coalgebra endowed with a coderivation $\CD$ of degree +1, which squares to zero: $\CD^2=0$.
\end{definition}
As an example, consider the graded symmetric algebra 
\begin{equation}
 S(V):=\FR\oplus V\oplus (V\odot V)\oplus V^{\odot 3}\oplus \ldots
\end{equation}
of a graded vector space $V$, where $\odot$ denotes the graded symmetric tensor product. We define a coproduct $\Delta(v)=v\otimes 1+1\otimes v$ for $v\in V$, which is extended to a coproduct on $S(V)$ according to 
\begin{equation}
\begin{aligned}
 &\Delta(v_1\odot \ldots\odot v_n):=\\
 &~~~\sum_{1\leq i\leq n-1} \sum_\sigma \eps(\sigma;v_1,\ldots,v_n)(v_{\sigma(1)}\odot v_{\sigma(2)}\odot\ldots\odot v_{\sigma(i)})\otimes(v_{\sigma(i+1)}\odot \ldots \odot v_{\sigma(n)})~.
\end{aligned}
\end{equation}
Here, $\sigma$ runs over all $(i,n-i)$-unshuffles, i.e.\ permutations of $\{1,2,\ldots,n\}$ such that the first $i$ and the last $n-i$ elements are ordered, and $\eps(\sigma)$ is the Koszul sign of the unshuffle, defined via
\begin{equation}
 v_1\otimes \ldots \otimes v_n=\eps(\sigma;v_1,\ldots,v_n)(v_{\sigma(1)}\otimes v_{\sigma(2)}\otimes\ldots \otimes v_{\sigma(n)})~.
\end{equation}
Together with this coproduct, $S(V)$ is a cocommutative, coassociative graded coalgebra. 

Given a coderivation $\CD$ on $S(V)$, we can project the image of $\CD$ onto $V$, resulting in linear maps $\CD^1$. The co-Leibniz rule \eqref{eq:co-Leibniz} then implies that
\begin{equation}\label{eq:implication_co_Leibniz}
\begin{aligned}
 &\CD(v_1\odot \ldots\odot v_n):=\CD^1(v_1\odot \ldots \odot v_n)+\\ 
 &~~\sum_{1\leq i\leq n-1} \sum_\sigma \eps(\sigma;v_1,\ldots,v_n)\CD^1(v_{\sigma(1)}\odot v_{\sigma(2)}\odot\ldots\odot v_{\sigma(i)})\odot v_{\sigma(i+1)}\odot \ldots \odot v_{\sigma(n)}~,
\end{aligned}
\end{equation}
where $\sigma$ runs again over all $(i,n-i)$-unshuffles. This equation shows that $\CD$ is completely defined by $\CD^1$.

Important for us is the following specialization. The {\em reduced symmetric algebra}
\begin{equation}
 \bar{S}(V):=V\oplus (V\odot V)\oplus V^{\odot 3}\oplus \ldots
\end{equation}
of a graded vector space $V$ also carries naturally the structure of a cocommutative, coassociative graded coalgebra with the reduced coproduct $\bar \Delta(v)=\Delta(v)-v\otimes 1-1\otimes v$, which extends to
\begin{equation}
\begin{aligned}
 \bar \Delta&(v_1\odot v_2\odot\ldots\odot v_n)=\\&\sum_{1\leq i\leq n-1} \sum_\sigma \eps(\sigma;v_1,\ldots,v_n)(v_{\sigma(1)}\odot v_{\sigma(2)}\odot\ldots\odot v_{\sigma(i)})\otimes (v_{\sigma(i+1)}\odot \ldots \odot v_{\sigma(n)})~,
\end{aligned}
\end{equation}
where $\sigma$ runs again over all $(i,n-i)$-unshuffles. Coderivations and codifferentials can now be defined on $\bar{S}(V)$ by using the reduced coproduct $\bar{\Delta}$.

\bibliography{bigone}

\bibliographystyle{latexeu}

\end{document}